\RequirePackage[l2tabu,orthodox]{nag}
\documentclass
[11pt,letterpaper]
{article} 

\usepackage[dvipsnames]{xcolor}

\usepackage[notes=false,later=false,camera=false,draft=false]{dtrt}
\usepackage[utf8]{inputenc}
\usepackage{etex}
\usepackage{ stmaryrd }
\usepackage{xspace,enumerate}
\usepackage[T1]{fontenc}
\usepackage[full]{textcomp}
\usepackage[american]{babel}
\usepackage{mathtools}

\usepackage{amsthm}
\usepackage{empheq}
\usepackage{thm-restate}

   \usepackage{hyperref}
      
\hypersetup{
colorlinks=true,
urlcolor=Cerulean,
linkcolor=RoyalBlue,
citecolor=OliveGreen,
linktocpage=true,
}
\usepackage[capitalise,nameinlink]{cleveref}
\crefname{lemma}{Lemma}{Lemmas}
\crefname{fact}{Fact}{Facts}
\newcommand{\colorconstraints}{\text{Color Constraints}}
\crefname{colorconstraints}{(color constraints)}{Color Constraints}
\crefformat{colorconstraints}{#2\colorconstraints#3}
\crefname{indsetconstraints}{(indset constraints)}{IndSet Constraints}
\crefformat{indsetconstraints}{#2$\mathsf{IndSet\ Axioms}$#3}
\crefname{theorem}{Theorem}{Theorems}
\crefname{mtheorem}{Theorem}{Theorems}
\crefname{corollary}{Corollary}{Corollaries}
\crefname{claim}{Claim}{Claims}
\crefname{example}{Example}{Examples}
\crefname{algorithm}{Algorithm}{Algorithms}
\crefname{problem}{Problem}{Problems}
\crefname{definition}{Definition}{Definitions}
\usepackage{paralist}
\usepackage{turnstile}
\usepackage{mdframed}
\usepackage{tikz}
\usepackage{caption}
\DeclareCaptionType{Algorithm}
\usepackage{newfloat}
\newtheorem{theorem}{Theorem}[section]
\newtheorem{mtheorem}{Theorem}
\newtheorem*{theorem*}{Theorem}

\newtheorem*{proposition*}{Proposition}
\newtheorem{lemma}[theorem]{Lemma}
\newtheorem*{lemma*}{Lemma}

\newtheorem*{conjecture*}{Conjecture}
\newtheorem{fact}[theorem]{Fact}
\newtheorem*{fact*}{Fact}

\newtheorem*{hypothesis*}{Hypothesis}

\theoremstyle{definition}
\newtheorem{definition}[theorem]{Definition}
\newtheorem*{definition*}{Definition}

\newtheorem{algorithm}[theorem]{Algorithm}

\theoremstyle{remark}
\newtheorem{claim}[theorem]{Claim}
\newtheorem*{claim*}{Claim}
\newtheorem{remark}[theorem]{Remark}
\newtheorem*{remark*}{Remark}

\newtheorem*{observation*}{Observation}

\usepackage[
letterpaper,
top=1.2in,
bottom=1.2in,
left=1in,
right=1in]{geometry}
\usepackage{newpxtext} 
\usepackage{textcomp} 
\usepackage[varg,bigdelims]{newpxmath}
\usepackage[scr=rsfso]{mathalfa}
\usepackage{bm} 
\let\mathbb\varmathbb
\usepackage{microtype}
\definecolor{petergreen}{rgb}{0, 0.75, 0}

\allowdisplaybreaks
\newcommand{\FormatAuthor}[3]{
\begin{tabular}{c}
#1 \\ {\small\texttt{#2}} \\ {\small #3}
\end{tabular}
}

\newcommand{\keywords}[1]{\bigskip\par\noindent{\footnotesize\textbf{Keywords\/}: #1}}

\newcommand{\R}{{\mathbb R}}
\newcommand{\N}{{\mathbb N}}
\newcommand{\norm}[1]{\lVert #1 \rVert}

\newcommand{\abs}[1]{\lvert #1 \rvert}

\newcommand{\floor}[1]{\lfloor #1 \rfloor}

\newcommand{\eps}{\varepsilon}

\newcommand{\E}{{\mathbb E}}

\newcommand{\Bits}{\{0,1\}}

\newcommand{\Fits}{\{-1,1\}}

\newcommand{\Var}{\mathrm{Var}}

\newcommand{\poly}{\mathrm{poly}}
\DeclareMathOperator*{\val}{\mathrm{val}}

\newcommand{\mper}{\,.}
\newcommand{\mcom}{\,,}

\DeclareMathOperator*{\polylog}{\mathrm{polylog}}

\newcommand{\Code}{\mathcal{C}}

\newcommand{\Dec}{\mathrm{Dec}}

\newcommand{\defeq}{\coloneqq}

\newcommand{\AND}{\mathrm{AND}}

\renewcommand{\Xi}{\xi}

\newcommand\MYcurrentlabel{xxx}
\newcommand{\MYstore}[2]{%
  \global\expandafter \def \csname MYMEMORY #1 \endcsname{#2}%
}
\newcommand{\MYload}[1]{%
  \csname MYMEMORY #1 \endcsname%
}
\newcommand{\MYnewlabel}[1]{%
  \renewcommand\MYcurrentlabel{#1}%
  \MYoldlabel{#1}%
}
\newcommand{\MYdummylabel}[1]{}
\newcommand{\torestate}[1]{%
  \let\MYoldlabel\label%
  \let\label\MYnewlabel%
  #1%
  \MYstore{\MYcurrentlabel}{#1}%
  \let\label\MYoldlabel%
}
\newcommand{\restatetheorem}[1]{%
  \let\MYoldlabel\label
  \let\label\MYdummylabel
  \begin{theorem*}[Restatement of \cref{#1}]
    \MYload{#1}
  \end{theorem*}
  \let\label\MYoldlabel
}
\newcommand{\restatelemma}[1]{%
  \let\MYoldlabel\label
  \let\label\MYdummylabel
  \begin{lemma*}[Restatement of \cref{#1}]
    \MYload{#1}
  \end{lemma*}
  \let\label\MYoldlabel
}
\newcommand{\restateprop}[1]{%
  \let\MYoldlabel\label
  \let\label\MYdummylabel
  \begin{proposition*}[Restatement of \cref{#1}]
    \MYload{#1}
  \end{proposition*}
  \let\label\MYoldlabel
}
\makeatletter
\newcommand{\subalign}[1]{%
  \vcenter{%
    \Let@ \restore@math@cr \default@tag
    \baselineskip\fontdimen10 \scriptfont\tw@
    \advance\baselineskip\fontdimen12 \scriptfont\tw@
    \lineskip\thr@@\fontdimen8 \scriptfont\thr@@
    \lineskiplimit\lineskip
    \ialign{\hfil$\m@th\scriptstyle##$&$\m@th\scriptstyle{}##$\hfil\crcr
      #1\crcr
    }%
  }%
}
\makeatother

\newcommand{\restatefact}[1]{%
  \let\MYoldlabel\label
  \let\label\MYdummylabel
  \begin{fact*}[Restatement of \prettyref{#1}]
    \MYload{#1}
  \end{fact*}
  \let\label\MYoldlabel
}
\newcommand{\restate}[1]{%
  \let\MYoldlabel\label
  \let\label\MYdummylabel
  \MYload{#1}
  \let\label\MYoldlabel
}

\begin{document}

\title{A $k^{\frac{q}{q-2}}$ Lower Bound for Odd Query Locally Decodable Codes from Bipartite Kikuchi Graphs}

\author{
\begin{tabular}[h!]{cc}
      \FormatAuthor{Oliver Janzer}{oliver.janzer@epfl.ch}{EPFL} &
      \FormatAuthor{Peter Manohar\thanks{This material is based upon work supported by the National Science Foundation
under Grant No.\ DMS-1926686.}}{pmanohar@ias.edu}{The Institute for Advanced Study}
\end{tabular}
} %
\date{\today}

\maketitle

\begin{abstract}
A code $\Code \colon \Bits^k \to \Bits^n$ is a $q$-query locally decodable code ($q$-LDC) if one can recover any chosen bit $b_i$ of the message $b \in \Bits^k$ with good confidence by querying a corrupted string $\tilde{x}$ of the codeword $x = \Code(b)$ in at most $q$ coordinates. For $2$ queries, the Hadamard code is a $2$-LDC of length $n = 2^k$, and this code is in fact essentially optimal~\cite{KerenidisW04,GoldreichKST06}. For $q \geq 3$, there is a large gap in our understanding: the best constructions achieve $n = \exp(k^{o(1)})$, while prior to the recent work of~\cite{AlrabiahGKM23}, the best lower bounds were $n \geq \tilde{\Omega}(k^{\frac{q}{q-2}})$ for $q$ even and $n \geq \tilde{\Omega}(k^{\frac{q+1}{q-1}})$ for $q$ odd.

The recent work of~\cite{AlrabiahGKM23} used techniques from semirandom XOR refutation to prove a lower bound of $n \geq \tilde{\Omega}(k^3)$ for $q = 3$, thus achieving the ``$k^{\frac{q}{q-2}}$ bound'' for an odd value of $q$. However, their proof does not extend to any odd $q \geq 5$.
In this paper, we prove a $q$-LDC lower bound of $n \geq \tilde{\Omega}(k^{\frac{q}{q-2}})$ for \emph{any} odd $q$. Our key technical idea is the use of an imbalanced bipartite Kikuchi graph, which gives a simpler method to analyze spectral refutations of \emph{odd} arity XOR without using the standard ``Cauchy--Schwarz trick'' --- a trick that typically produces random matrices with nontrivially correlated entries and makes the analysis for odd arity XOR significantly more complicated than even arity XOR.
\end{abstract}
\thispagestyle{empty}
\clearpage
 \microtypesetup{protrusion=false}
  \tableofcontents{}
  \microtypesetup{protrusion=true}

\thispagestyle{empty}
\clearpage

\pagestyle{plain}
\setcounter{page}{1}

\section{Introduction}
A (binary) locally decodable code (LDC) $\Code \colon \Bits^k \to \Bits^n$ is an error correcting code that admits a local decoding algorithm --- for any message $b \in \Bits^k$ and any string $\tilde{x} \in \Bits^n$ obtained by corrupting the codeword $x = \Code(b)$ in a small constant fraction of coordinates, the local decoder is able to recover any bit $b_i$ of $b$ with good confidence while only reading a small number of coordinates of the corrupted codeword $\tilde{x}$. More formally, we say that $\Code$ is $(q, \delta, \eps)$-locally decodable if the decoder only reads at most $q$ bits of the corrupted string, and for any $\tilde{x}$ with Hamming distance $\Delta(x, \tilde{x}) \defeq \abs{\{u \in [n] : x_u \ne \tilde{x}_u\}} \leq \delta n$ and any input $i \in [k]$ to the decoder, the decoder recovers $b_i$ with probability at least $\frac{1}{2} + \eps$. Locally decodable codes were first formally defined in the work of~\cite{KatzT00}, although they were instrumental components in the earlier proof of the PCP theorem \cite{AroraS98, AroraLMSS98}, and have connections to complexity theory (see \cite{Trevisan04} and~\cite[Section 7]{Yekhanin12}). Example applications include worst-case to average-case reductions \cite{Trevisan04}, private information retrieval \cite{Yekhanin10}, secure multiparty computation \cite{IshaiK04}, derandomization \cite{DvirS05}, matrix rigidity \cite{Dvir10}, data structures \cite{Wolf09,ChenGW10}, and fault-tolerant computation \cite{Romashchenko06}.

A central question in coding theory is to determine the optimal blocklength $n$ of a $(q,\delta,\eps)$-LDC as a function of $k$, the length of the message, and $q$, the number of queries, in the regime where $q$ is constant and $\delta, \eps$ are also constant. The work of~\cite{KatzT00} shows that there are no $1$-query locally decodable codes unless $k$ is constant, so the first nontrivial setting of $q$ is $q = 2$, and they additionally prove a superlinear lower bound for all $q \geq 2$. For $2$-query locally decodable codes, we have an essentially complete understanding: the Hadamard code gives a $2$-LDC with $n = 2^k$, and the works of~\cite{KerenidisW04,GoldreichKST06} show a lower bound of $n \geq 2^{\Omega(k)}$, which is therefore tight up to a constant in the exponent.

Unlike the case of $q = 2$, for $q \geq 3$ there is a large gap between the best-known upper and lower bounds on $n$. The best-known upper bound, i.e., construction, comes from matching vector codes~\cite{Yekhanin08,Efremenko09}, and achieves, in the case of $q = 3$, a blocklength of $n = \exp(\exp(O(\sqrt{\log k \log \log k})))$. This is $2^{k^{o(1)}}$, i.e., subexponential in $k$, which is substantially smaller than the Hadamard code, the code of optimal length for $q = 2$. More generally, for any constant $q = 2^r$, the works of~\cite{Yekhanin08,Efremenko09} construct $q$-query locally decodable codes of length $n = \exp(\exp(O( (\log k)^{1/r} (\log \log k)^{1 - 1/r})))$, which has a similar qualitative subexponential behavior.

On the other hand, the known lower bounds for $q \geq 3$ are substantially weaker. The original work of~\cite{KatzT00} proves that a $q$-LDC has blocklength $n \geq \Omega(k^{\frac{q}{q-1}})$. This was later improved by the work of~\cite{KerenidisW04}, which showed that for even $q$, a $q$-LDC has blocklength $n \geq \tilde{\Omega}(k^{\frac{q}{q-2}})$. For odd $q$, they observe that a $q$-LDC is also a $(q+1)$-LDC where $q+1$ is now even, and so their lower bound for even $q$ trivially yields a bound of $n \geq \tilde{\Omega}(k^{\frac{q+1}{q-1}})$ for odd $q$.

The lower bounds of~\cite{KerenidisW04} remained, up to $\polylog(k)$ factors, the best lower bounds known until the recent work of~\cite{AlrabiahGKM23}, which used spectral methods developed in the work of~\cite{GuruswamiKM22} for refuting constraint satisfaction problems to prove a lower bound of $n \geq \tilde{\Omega}(k^3)$ for $3$-LDCs. This lower bound of~\cite{AlrabiahGKM23} was the first improvement in any LDC lower bound by a $\poly(k)$ factor since the work of~\cite{KerenidisW04}, and achieves the ``$k^{\frac{q}{q-2}}$ bound'' established for even $q$ for the odd value of $q = 3$.
However, the proof of~\cite{AlrabiahGKM23} does not extend to any odd $q \geq 5$, and while the spectral method approach of~\cite{AlrabiahGKM23} was used in recent work of~\cite{KothariM23} (and follow-ups \cite{Yankovitz24,AlrabiahG24,KothariM24}) to prove an exponential lower bound for $3$-query locally \emph{correctable} codes (LCCs)\footnote{These lower bounds are for both linear $3$-LCCs (\cite{KothariM23,Yankovitz24,AlrabiahG24}) and nonlinear $3$-LCCs (\cite{KothariM24}), as well as the more specialized case of ``design $3$-LCCs'' (\cite{Yankovitz24,KothariM24}). The work of~\cite{AlrabiahG24} also proves a slightly stronger lower bound of $n \geq \tilde{\Omega}(k^{\frac{q-1}{q-3}})$ for $q$-LCCs for odd $q \geq 5$.}  --- a stronger variant of an LDC where the decoder must additionally be able to correct any bit $x_u$ of the uncorrupted codeword --- there have been no improvements in $q$-LDC lower bounds since~\cite{AlrabiahGKM23}. In particular, because the proof of~\cite{AlrabiahGKM23} does not extend to all odd $q \geq 5$, the best known lower bounds for $q$-LDCs are: \begin{inparaenum}[(1)] \item $n \geq \tilde{\Omega}(k^{\frac{q}{q-2}})$, if $q$ is even or $q = 3$, and \item $n \geq \tilde{\Omega}(k^{\frac{q+1}{q-1}})$ if $q \geq 5$ is odd\end{inparaenum}. Thus, a natural question to ask is: \emph{can we close this gap and prove a $n \geq \tilde{\Omega}(k^{\frac{q}{q-2}})$ lower bound for $q$-LDCs for all constant $q$?}

As the main result of this paper, we prove the following theorem, which establishes this lower bound.
\begin{mtheorem}
\label{mthm:oddqldc}
Let $\Code \colon \Bits^k \to \Bits^n$ be a $(q, \delta, \eps)$-locally decodable code with $q \geq 3$ and $q$ odd. Then, $k \leq O_q(n^{1 - \frac{2}{q}} \eps^{-6 - \frac{2}{q}} \delta^{-2 - \frac{2}{q}} \log n)$. In particular, if $q,\delta, \eps$ are constants, then $n \geq \Omega\left(\left(k/\log k\right)^{\frac{q}{q-2}} \right)$.
\end{mtheorem}
The main contribution of \cref{mthm:oddqldc} is that it improves the $q$-LDC lower bound for $q \geq 5$ from $n \geq \tilde{\Omega}(k^{\frac{q+1}{q-1}})$ to $n \geq \tilde{\Omega}(k^{\frac{q}{q-2}})$, which is a $\poly(k)$ factor improvement. However, we additionally note that for $q = 3$, \cref{mthm:oddqldc} has a better dependence on the lower order terms of $\log k, \delta, \eps$ hidden in the $\tilde{\Omega}(\cdot)$ as compared to the result of~\cite{AlrabiahGKM23}, which showed the weaker bound of $n \geq \Omega\left(\frac{\eps^{32} \delta^{16} k^3}{ \log^6 k} \right)$. This improvement in the $\log k$ factor for $q = 3$ was also recently obtained by~\cite{HsiehKMMS24} for the special case of linear codes. We note that \cref{mthm:oddqldc} additionally implies lower bounds for $q$-LDCs over larger alphabets via a known generic reduction to the binary case; see~\cite[Lemma 2]{KerenidisW04} and~\cite[Appendix A]{AlrabiahGKM23}.

As stated in~\cite{AlrabiahGKM23}, the proof techniques of~\cite{AlrabiahGKM23} extend to any $q \geq 5$ \emph{under the additional assumption that the code $\Code$ satisfies some extra nice regularity properties.}\footnote{In fact, as we mention in \cref{sec:agkmodd}, it turns out that the proof strategy of~\cite{AlrabiahGKM23} extends easily to the case of $q = 5$ without the need for any additional assumption, contrary to what is claimed in~\cite{AlrabiahGKM23}. The fact that this was missed by~\cite{AlrabiahGKM23} appears to be an oversight on their part. Nonetheless, their approach does break down for $q \geq 7$.} This condition arises for $q$ odd but not $q$ even for the following reason: in the proof, one defines a matrix where we would like to ``evenly split'' a degree-$q$ monomial $x_{v_1} \dots x_{v_q}$ across the rows and columns of the matrix. When $q$ even this is possible, as we can divide the monomial into two ``halves''. This property allows us to define a matrix with independent bits of randomness and obtain a somewhat simple proof. However, for $q$ odd, the best possible split is of course $(\frac{q-1}{2}, \frac{q+1}{2})$, which is slightly imbalanced. To handle this issue of imbalance, \cite{AlrabiahGKM23} uses the standard ``Cauchy--Schwarz trick'' developed in the context of spectral refutation algorithms for constraint satisfaction problems precisely to tackle this issue of imbalance. The ``Cauchy--Schwarz trick'' produces degree $2(q-1)$ monomials, which are even, but at the cost of making the randomness \emph{dependent}. This dependence in the randomness makes the analysis for odd $q$ considerably more technical than the more straightforward analysis for even $q$, and is where the aforementioned additional assumption on the code is needed. The fact that even $q$ is substantially more easier to handle from a technical perspective compared to odd $q$ is a reoccurring theme in the CSP refutation literature that has appeared in many prior works~\cite{CojaGL07,AllenOW15,BarakM16,RaghavendraRS17, AbascalGK21, GuruswamiKM22, HsiehKM23, KothariM23}.

Our main technical contribution is the introduction of an \emph{imbalanced} matrix, or equivalently a \emph{bipartite} graph for odd arity instances, that allows us to refute certain odd arity instances \emph{without} using the Cauchy--Schwarz trick. 
By using an imbalanced matrix and bypassing the Cauchy--Schwarz trick, we maintain independence in the randomness in our matrix (as opposed to introducing correlations), which makes the proof considerably simpler. The simpler proof has the additional advantage that, as mentioned earlier, in the case of $q = 3$ we can improve on the lower bound of~\cite{AlrabiahGKM23} by a $\log k \cdot \poly(1/\eps, 1/\delta)$ factor. Our use of a bipartite graph in the proof is perhaps surprising, as it is contrary to the conventional wisdom that symmetric matrices (i.e., normal, non-bipartite graphs) ought to produce the best spectral certificates. Indeed, the purpose of the ``Cauchy--Schwarz trick'' is to turn an odd arity instance into an even arity instance so that we can represent it with a balanced matrix, with the expectation (that is true in many cases) that the balanced matrix will produce a better spectral certificate.

The bipartite graph that we produce is a Kikuchi graph, i.e., a carefully chosen induced subgraph of a Cayley graph on the hypercube. As we note in \cref{rem:bipartitekikuchi}, bipartite Kikuchi graphs have appeared in prior works, namely~\cite{Yankovitz24, KothariM24}. However, the graphs in those works can be converted to non-bipartite graphs via a straightforward application of the Cauchy--Schwarz inequality, and so they are not ``inherently bipartite''. To our knowledge, our work is the first work to use such a graph that is \emph{inherently} bipartite, meaning that no easy conversion to a non-bipartite graph via the Cauchy--Schwarz inequality appears to exist.

\parhead{Concurrent work.} In concurrent work, \cite{BasuHKL25} also proves a $n \geq \tilde{\Omega}(k^{\frac{q}{q-2}})$ lower bound for $q$-query locally decodable codes for odd $q$. Their bound is slightly weaker compared to \cref{mthm:oddqldc}, as it has a worse $\log n$ dependence. Namely, for constant $\eps, \delta$, \cite{BasuHKL25} proves that $k \leq O(n^{1 - 2/q} (\log n)^{4})$ for nonlinear codes and $k \leq O(n^{1 - 2/q} (\log n)^{2})$ for linear codes, whereas \cref{mthm:oddqldc} proves that $k \leq O(n^{1 - 2/q} \log n)$ for both nonlinear and linear codes, which is a stronger bound in both cases.

\section{Proof Overview}
\label{sec:overview}
In this section, we give an overview of our proof and the techniques that we use. We will start with a thorough summary of the approach of~\cite{AlrabiahGKM23}, first for the (easier) case of even $q$, and then for the more involved case of $q = 3$. Then, we will explain our approach using bipartite Kikuchi graphs.

For the purpose of this overview, we will assume for simplicity that the code $\Code$ is linear, although we note that the proof for nonlinear codes does not change in any meaningful way.

\subsection{The approach of~\cite{AlrabiahGKM23} for even $q$}
\label{sec:agkmeven}
By standard reductions (\cref{fact:LDCnormalform}), for any linear $q$-LDC $\Code \colon \Bits^k \to \Bits^n$, there exist $q$-uniform hypergraph matchings (\cref{def:hypergraphs}) $H_1, \dots, H_k$ on the vertex set $[n]$, each with $\abs{H_i} = \delta n$ hyperedges, such that for each $i \in [k]$ and each hyperedge $C \in H_i$, it holds that $\sum_{v \in C} x_v = b_i$ when $x = \Code(b)$. One should think of the hyperedges $H_i$ as the set of query sets that the decoder may query on input $i$. That is, the decoder, when given input $i \in [k]$, simply chooses a random $C \gets H_i$, reads $x \vert_C$, and then outputs $\sum_{v \in C} x_v$. The linear constraints $\sum_{v \in C} x_v = b_i$ that are satisfied by $x = \Code(b)$ for all $b \in \Bits^k$ imply that the decoder succeeds in correctly recovering $b_i$ with probability $1$ on an uncorrupted codeword.

Switching from $\Bits$-notation to $\Fits$-notation via the map $0 \mapsto 1$ and $1 \mapsto -1$, the above implies that for any $b \in \Fits^k$, the system of constraints given by $\prod_{v \in C} x_v = b_i$ for each $i \in [k]$ and $C \in H_i$ is satisfiable, with $x = \Code(b)$ being a satisfying assignment. This implies that the degree-$q$ polynomial $\Phi_b(x) \defeq \sum_{i = 1}^k b_i \sum_{C \in H_i} \prod_{v \in C} x_v$ has value $\val(\Phi_b) \defeq \max_{x \in \Fits^n} \Phi_b(x) = \sum_{i = 1}^k \abs{H_i} = \delta n k$ for all $b \in \Fits^k$. Indeed, by setting $x = \Code(b)$, we have that $\prod_{v \in C} x_v = b_i$ for each $i \in [k]$ and $C \in H_i$, and so $\Phi_b(\Code(b)) = \sum_{i = 1}^k \abs{H_i} = \delta nk$. Thus, to prove a lower bound on $n$, it suffices to show that for any $H_1, \dots, H_k$ of size $\delta n$, if $n$ is too small, then there exists $b \in \Fits^k$ such that $\val(\Phi_b) < \delta nk$. 

We do this by bounding $\E_{b \gets \Fits^k}[\val(\Phi_b)]$ using a \emph{spectral certificate}. The certificate is as follows. First, we define the Kikuchi matrix/graph $A_C$.
\begin{definition}[Basic Kikuchi matrix/graph for $q$ even]
\label{def:basickikuchi}
Let $\ell \geq q$ be a positive integer (which we will set to $n^{1 - 2/q}$ eventually), and let $C \in {[n] \choose q}$. Let $A_C$ be the matrix with rows and columns indexed by sets $S,T \in {[n] \choose \ell}$ where $A_C(S,T) = 1$ if $S \oplus T = C$ and $A_C(S,T) = 0$ otherwise. Here, $S \oplus T$ denotes the symmetric difference of $S$ and $T$, which is $\{u : (u \in S \wedge u \notin T) \vee (u \notin S \wedge u \in T)\}$. We will at times refer to the matrix $A_C$ as a graph (where we identify $A_C$ with the graph with adjacency matrix $A_C$), and then we will refer to the nonzero entries $(S,T)$ as edges.
\end{definition}
Notice that the condition that $S \oplus T = C$ is equivalent to the existence of a partition $C = C_1 \cup C_2$ into two sets of size $\frac{q}{2}$ such that $S \cap C = C_1$, $T \cap C = C_2$, and $S \setminus C_1 = T \setminus C_2$. Namely, $S \oplus T =C$ if and only if we can split the hyperedge $C$ \emph{evenly} across $S$ and $T$ --- notice that here we crucially require that $q$ is even for the matrix $A_C$ to have a single nonzero entry!

The matrix $A_C$ is a Kikuchi matrix, first introduced in the work of~\cite{WeinAM19} for the problem of tensor PCA, and has the following nice properties: \begin{inparaenum}[(1)] \item the matrix $A_C$ has exactly $D = {q \choose q/2} {n - q \choose \ell - q/2}$ nonzero entries, and \item for each $x \in \Fits^n$, letting $z \in \Fits^{{n \choose \ell}}$ denote the vector where $z_S \defeq \prod_{v \in S} x_v$, we have $z^{\top} A_C z = D \prod_{v \in C} x_v$\end{inparaenum}. These two properties allow us to use the $A_C$'s as a ``basis'' to express any homogeneous degree-$q$ polynomial in variables $x \in \Fits^n$ as a quadratic form on a linear combination of $A_C$'s. Namely, if we let $A_i \defeq \sum_{C \in H_i} A_C$ and $A = A_b \defeq \sum_{i = 1}^k b_i A_i$ (mimicking the definition of $\Phi_b$), then we have $z^{\top} A z = D \Phi_b(x)$ for every $x \in \Fits^n$, where $z$ is defined as before. We can thus express $\Phi_b(x)$ as a quadratic form on the matrix $A$, and so we have shown that $\delta n k \leq \val(\Phi_b) \leq \norm{A}_2 \cdot {n \choose \ell}$. So, to finish the proof, it remains to bound $\E_b[\norm{A}_2]$.

As each $b_i$ is chosen independently from $\Fits$, the matrix $A = \sum_{i = 1}^k b_i A_i$ is a Matrix Rademacher series, and so we can bound its spectral norm using the Matrix Khintchine inequality (\cref{fact:matrixkhintchine}). This implies that $\E_b[\norm{A}_2] \leq O(\Delta \sqrt{k \ell \log n})$, where $\Delta$ is the maximum number of $1$'s in a row any of the $A_i$'s. As the $A_i$'s are symmetric matrices with entries in  $\Bits$, we can view them as adjacency matrices of graphs. With this perspective, $\Delta$ is simply the maximum degree of a vertex $S$ in any of the $A_i$'s. 

The maximum degree $\Delta$ can never be smaller than the average degree in an $A_i$, which is $\delta n D / {n \choose \ell}$. Thus, if each $A_i$ is approximately regular, so that the maximum degree is on the same order of magnitude as the average degree, then we would be able to conclude that
\begin{flalign*}
&\delta n k D \leq {n \choose \ell} \E_b[\norm{A}_2] \leq {n \choose \ell} \cdot \frac{\delta n D}{ {n \choose \ell}} \cdot O(\sqrt{k \ell \log n}) \\
&\implies k \leq O(\ell \log n) \mper
\end{flalign*}
Ideally, we would like to take $\ell$ to be as small as possible to now get the best possible bound on $k$. However, the maximum degree $\Delta$ is always at least $1$, and so for $\Delta$ to be on the same order of magnitude as the average degree, we must have average degree $\geq \Omega(1)$. The average degree is $\delta n D/ {n \choose \ell}$, which a simple calculation shows is roughly $\delta n \left(\frac{\ell}{n}\right)^{q/2}$, and so we need to take $\ell \geq n^{1 - 2/q}$. This means that our potential bound is $k \leq O(n^{1 - 2/q} \log n)$, i.e., $n \geq \tilde{\Omega}(k^{\frac{q}{q-2}})$, as desired.

\parhead{Finding an approximately regular subgraph: row pruning.} We have shown that if each graph $A_i$ is approximately regular, meaning that its maximum degree is on the same order of magnitude as its average degree, then we can prove $n \geq \tilde{\Omega}(k^{\frac{q}{q-2}})$. Unfortunately, the graph $A_i$ is not approximately regular, \emph{even though the underlying hypergraph $H_i$ is a matching}, i.e., $H_i$ is as ``regular'' as possible. A naive way to try to enforce this ``approximately regular'' property is to simply remove all vertices $S$ with large degree in $A_i$ (along with their adjacent edges), producing a new graph $B_i$ with max degree bounded by $O(1)$ times the average degree of $A_i$. However, for a general graph, this deletion process may delete most (or all!) of the edges, resulting in a considerable drop in the average degree. So, the resulting graph $B_i$ need not be approximately regular. Crucially, \emph{because $H_i$ is a matching}, we can show that this process in fact only deletes a $o(1)$-fraction of the edges, and so the average degree is essentially unchanged. This means that the graph $B_i$ \emph{is} indeed approximately regular, and so we can use the $B_i$'s in place of the $A_i$'s to finish the proof.

The above vertex/edge deletion step is typically called the ``row pruning'' step, so-named because it prunes rows (and columns) of the matrix $A_i$, and has appeared in many prior works that analyze spectral norms of Kikuchi matrices. While at first glance this step may appear to be a mere technical annoyance, it is in fact the most critical part of the entire proof. In this case of the above proof, we note that this is the only step that uses that the $H_i$'s are matchings, and if the $H_i$'s are not matchings then the lower bound is clearly false. In fact, in the entire proof above, one should view all the steps up until the row pruning step as \emph{generic} and dictated by the polynomial $\Phi_b$ whose value we wish to bound, and the row pruning step is the key part of the proof that determines if the approach succeeds in obtaining a strong enough bound on $\E_b[\val(\Phi_b)]$.

As observed in~\cite[Appendix B]{AlrabiahGKM23}, one can also view the above proof as giving a reduction from a $q$-LDC to a $2$-LDC for even $q$. In this viewpoint, the row pruning step is the crucial part of the proof that shows that the object produced by the reduction is in fact a $2$-LDC.

\subsection{The approach of~\cite{AlrabiahGKM23} for $q = 3$ and why it fails for odd $q \geq 7$}
\label{sec:agkmodd}
We now recall the approach of~\cite{AlrabiahGKM23} for $q = 3$ and explain why its natural generalization to odd $q \geq 7$ fails. As briefly mentioned earlier, the reason the previous proof does not succeed for odd $q$ is because the matrix $A_C$ has no nonzero entries when $\abs{C}$ is odd. This is because the row sets $S$ and the column sets $T$ have exactly the same size $\ell$; the Kikuchi graph would have nonzero entries if we, e.g., simply allowed $\abs{T} = \ell + 1$. Namely, we can make the following definition.
\begin{definition}[Naive imbalanced Kikuchi matrix for odd $q$]
\label{def:naiveoddkikuchi}
Let $q$ be odd and let $C \in {[n] \choose q}$. Let $\ell \geq q$ be a positive integer. Let $A'_C$ be the matrix with rows indexed by sets $S \in {[n] \choose \ell}$ and columns indexed by sets $T \in {[n] \choose \ell+1}$, where $A'_C(S,T) = 1$ if $S \oplus T = C$ and otherwise $A'_C(S,T) = 0$.
\end{definition}
Analogously to \cref{def:basickikuchi}, if the $(S,T)$-th entry of $A'_C$ is nonzero, then $\abs{S \cap C} = \frac{q - 1}{2}$ and $\abs{T \cap C} = \frac{q+1}{2}$. This imbalance causes the average left degree of $A'_i \defeq \sum_{C \in H_i} A'_C$ to be roughly  $\delta n \left(\frac{\ell}{n}\right)^{(q-1)/2}$, while the average right degree is $\delta n \left(\frac{\ell}{n}\right)^{(q+1)/2}$. In order for the row pruning step to have any hope for success, we need both of these quantities to be at least $1$, which requires taking $\ell \geq n^{1 - \frac{2}{q+1}}$. In fact, using this asymmetric matrix precisely reproduces the $n \geq \tilde{\Omega}(k^{\frac{q+1}{q-1}})$ bound.

The first key step in the proof of~\cite{AlrabiahGKM23} for $q = 3$ is to use the ``Cauchy--Schwarz trick'' from the CSP refutation literature: we construct a new system of constraints by first taking two constraints $ x_u x_{v_1} x_{v_2} = b_i$ and $x_u x_{w_1} x_{w_2} = b_j$ \emph{that both contain the same variable $x_u$} and then we multiply them together to derive a new constraint $x_{v_1} x_{v_2} x_{w_1} x_{w_2} = b_i b_j$, using that $x_u^2 = 1$ since $x_u \in \Fits$.
Crucially, since we started with $\Omega(nk)$ constraints, the number of derived constraints is $\Omega(nk^2)$, and the arity of the derived constraint is $4$ (or more generally $2(q-1)$),\footnote{The degree may be smaller if the constraints share at least $2$ variables, but this would reduce the degree further and so it is only ``better'' for us. There are several simple ways to handle this issue, but we will ignore this technicality for the purpose of simplifying this proof overview.} which is now even. We thus define a new polynomial $\Psi_b$ of even degree for the derived instance:
\begin{equation*}
\Psi_b(x) \defeq \sum_{i \ne j} b_i b_j \sum_{u \in [n]} \sum_{\substack{(u, v_1, v_2) \in H_i \\ (u, w_1, w_2) \in H_j}}   x_{v_1} x_{v_2} x_{w_1} x_{w_2} \mper
\end{equation*}
 A simple application of the Cauchy--Schwarz inequality relates $\val(\Phi_b)$ and $\val(\Psi_b)$, and hence this derivation process is typically called the ``Cauchy--Schwarz trick''. The main drawback is that in the derived constraints, the ``right-hand sides'' are products $b_i b_j$, and we have introduced correlations in the right-hand sides.

We can now use the Kikuchi graphs $A_{C}$ (\cref{def:basickikuchi}) for each \emph{derived constraint} $C$, as the derived constraints have even arity. However, we will make one small, but crucial change. For a derived constraint $x_{v_1} x_{v_2} x_{w_1} x_{w_2}$, where $v_1, v_2$ ``come from'' one hyperedge $(u, v_1, v_2)$ and $w_1, w_2$ ``come from'' the other hyperedge $(u, w_1, w_2)$, we view this constraint as two pairs $(\{v_1, v_2\}, \{w_1, w_2\})$, and for an edge $(S, T)$ in the graph $A_{(\{v_1, v_2\}, \{w_1, w_2\})}$, we require that $S$ contains one element from each of $\{v_1, v_2\}$ and $\{w_1, w_2\}$, and that $T$ contains the other element from each. That is to say, we \emph{evenly split} the variables from the underlying (original) hyperedges across the row set $S$ and column set $T$. The fact that we split elements evenly is crucial for the row pruning step that we will discuss shortly.

With the above definition of the matrix $A_{(\{v_1, v_2\}, \{w_1, w_2\})}$, we can then make the following definitions. First, we partition $[k]$ randomly into two sets $L \cup R$, with $\abs{L} \geq \frac{k}{2}$ without loss of generality. Then, we let $A_{i,j} \defeq \sum_{u \in [n]} \sum_{\substack{(u, v_1, v_2) \in H_i \\ (u, w_1, w_2) \in H_j}} A_{(\{v_1, v_2\}, \{w_1, w_2\})}$, $A_i \defeq \sum_{j \in R} b_j A_{i,j}$ and $A = \sum_{i \in L} b_i A_i$. The random partition of $[k]$ into $L \cup R$ is a nice trick used in~\cite{AlrabiahGKM23} that makes the matrix $A$ be the sum of mean $0$ independent random matrices. At this point, we can now take $\ell = n^{1 - 2/q} = n^{1/3}$ and apply similar steps as done in the case of $q$ even to finish the proof, provided that the ``approximately regular'' condition can be made to hold for each graph $A_i$, i.e., the row pruning step succeeds.

\parhead{Finding an approximately regular subgraph: row pruning.} Let us now discuss the row pruning step for the matrices $A_i$. Unlike in the even case, the constraint hypergraph that defines the matrix $A_i$ is no longer a matching. Instead, the edges in the graph $A_i$ ``come from'' tuples $(u, \{v_1, v_2\}, \{w_1, w_2\})$ where $(u, v_1, v_2) \in H_i$ and $(u, w_1, w_2) \in \cup_{j \in R} H_j$ --- here, $u$ is the shared variable that is ``canceled'' by multiplying the two constraints together. To find an approximately regular subgraph of $A_i$, intuitively we need to show that a typical vertex has degree roughly equal to the average degree. We can try to understand how concentrated the degrees are in $A_i$ by computing the variance of $\deg_i(S)$, the degree of $S$ in $A_i$, when $S$ is chosen uniformly at random (see \cref{lem:regularmoments} for a formal calculation that is closely related). Here, it is crucial that we have split the uncanceled variables $\{v_1, v_2\}$ of the hyperedge $(u, v_1, v_2)$ evenly across $S$ and $T$ because if we had not, then any set $S$ that contains both $\{v_1, v_2\}$ should\footnote{Formally, it has degree at least the number of hyperedges $(u, w_1, w_2) \in \cup_{j \in R} H_j$ that contain the variable $u$, and this is typically $\Omega(k)$. For example, it is $\Omega(k)$ for every $u$ if the $H_i$'s are random hypergraph matchings.}
 have degree $\Omega(k)$. This is much larger than the average degree, which one can show is $n^{-1/3} k$, and happens with probability $n^{-1/3}$: high enough to dominate the variance.
 
In fact, even when we use the even split, $\Var(\deg_i(S))$ may still be too large. However, from the calculation of $\Var(\deg_i(S))$, we can extract the following natural combinatorial condition that, if satisfied, will make the variance small enough to finish the proof: we require that each \emph{pair} of variables $\{u,w\}$ appears in at most $d_2 \defeq (\ell/n)^{\frac{1}{2}} k = n^{-1/3} k$ hyperedges in $\cup_{j \in R} H_j$.
However, the hypergraph $\cup_{j \in R} H_j$ is a union of matchings --- it is not a matching itself ---  and so it is quite possible that there are pairs of variables $\{u,w\}$ that appear in, say, $\Omega(k)$ hyperedges in $\cup_{j \in R} H_j$.\footnote{Because the $H_j$'s are matchings and $\abs{R} \leq k$, even a single variable $u$ cannot appear in more than $k$ hyperedges. This is why we do not encounter a ``heavy singleton'' condition.} In fact, if many such ``heavy pairs'' $\{u,w\}$ exist, then we are unable to show that the graph $A_i$ has an approximately regular subgraph, and the above proof fails!

Nonetheless, the above proof still accomplishes something nontrivial. For $q = 3$, we obtain a proof that $k \leq \tilde{O}(n^{1/3})$ under the additional assumption that each pair $\{u,w\}$ of variables appears in at most $n^{-1/3} k$ hyperedges in $\cup_{j \in [k]} H_j$.\footnote{As we do not know $R$ in advance, we must impose a global condition on $\cup_{j \in [k]}H_j$ instead of $\cup_{j \in R} H_j$. However, as we expect $\abs{R}$ to be about $k/2$, this is also only off by a constant factor.} More generally, for larger odd $q$, we can show that $k \leq \tilde{O}(n^{1 - 2/q})$ under the additional assumption that for any set $Q$ of size $\abs{Q} = s$ where $2 \leq s \leq \frac{q+1}{2}$, the set $Q$ appears in at most $d_s \defeq (\ell/n)^{s - \frac{3}{2}} k = n^{-\frac{2s}{q} + \frac{3}{q}} k$ hyperedges in $\cup_{j \in [k]} H_j$.

\parhead{Removing the heavy pairs assumption.} The final step in the proof of~\cite{AlrabiahGKM23} is to remove this assumption by using the \emph{hypergraph decomposition} method of~\cite{GuruswamiKM22}. For each heavy pair $\{u,w\}$, we create a new ``big variable'' $p$ and replace all hyperedges $(u, w, v)$ with a new hyperedge $(p, v)$. Then, we create a new set of derived constraints by canceling the heavy pair variables $p$, resulting in a new degree-$2$ polynomial whose value we can then bound.\footnote{Formally, the proof of~\cite{AlrabiahGKM23} proceeds slightly differently and uses a bipartite graph, although it is equivalent to this.} So, if there are many heavy pairs, then we can produce a degree-$2$ polynomial, and otherwise we already win via the degree-$4$ polynomial.

\parhead{The hypergraph decomposition strategy fails for $q \geq 7$.} The above approach to handling heavy pairs suggests a natural strategy to handle larger heavy sets. Namely, let $d_s \defeq n^{-\frac{2s}{q} + \frac{3}{q}} k$ be the ``threshold for heavy sets $Q$ of size $s$'' that we found via the variance calculation. For each $2 \leq s \leq \frac{q+1}{2}$, we let $P_s$ denote the set of heavy $Q$'s of size $s$.
Then, for each heavy set $Q$, we can introduce a new variable $p$ and replace all hyperedges $C \in \cup_{j \in [k]} H_j$ containing $Q$ where $Q$ is the largest heavy set (ties broken arbitrarily) with the hyperedge $(p, C \setminus Q)$. This will produce, for each $i \in [k]$, a hypergraph $H_i^{(s)}$ where each hyperedge in $H^{(s)}_i$ has the form $(p, C')$ where $\abs{C'} = q - s$ and $p \in P_s$ is a new variable.

The derivation strategy of the ``Cauchy--Schwarz trick'' now suggests that we should, for each $s$, group the hypergraphs $H^{(s)}_1, \dots, H^{(s)}_k$ together and derive constraints by canceling the new variables $p$. Namely, we take two hyperedges $(p, C)$ and $(p, C')$ that use the same $p$  and combine them to produce the derived constraint $(C, C')$ that has arity $2(q-s)$. Once again, we can define an analogous Kikuchi matrix and attempt to complete the proof, and the success of this strategy is determined by whether or not the row pruning step goes through.

It turns out that, for $q = 5$, this simple generalization does indeed succeed. We suspect that this was perhaps missed by~\cite{AlrabiahGKM23} because the thresholds $d_s$ are rather delicate, and the proof breaks if we set, e.g., $d_2 = n^{-1/5} \cdot n^{1 - \frac{2}{5}} \log n$ as opposed to $n^{-1/5} k$ (recall that we expect $k$ to be about $n^{1 - \frac{2}{5}} \log n$ as this is the lower bound that we are shooting for).

However, for $q \geq 7$, this proof strategy fails. The first case that breaks is for $q = 7$ and $s = 4$, i.e., we have produced derived constraints $(C,C')$ of arity $6$ by canceling a heavy $4$-tuple. The issue is that the hypergraphs $H^{(4)}_1, \dots, H^{(4)}_k$, which have hyperedges containing $q - s = 7 - 4 = 3$ original variables from $[n]$, may still contain heavy pairs or triples. It turns out that, for any odd $q$, the cases of $s = 2$ and $s = 3$ never break, so this problem does not arise for $q = 5$ (recall that $2 \leq s \leq \frac{q+1}{2}$, which is $3$ when $q = 5$).

\ignore[recursive hypergraph decomp part]{
\subsection{Recursive hypergraph decomposition and Kikuchi graphs for partite XOR}
\label{sec:recursivedecomp}
While the strategy described above fails for $q \geq 7$, there is again a natural next step to try. The problem with, e.g., the case of $q = 7$ and $s = 4$, is that the hypergraphs $H^{(4)}_1, \dots, H^{(4)}_k$ may still contain heavy pairs or triples. So, we can simply recurse and decompose these hypergraphs again to produce hypergraphs $H^{(4,3)}_1, \dots, H^{(4,3)}_k$ and $H^{(4,2)}_1, \dots, H^{(4,2)}_k$. Here, hyperedges in $H^{(4,3)}_i$ have the form $(p^{(4)}, p^{(3)})$ where $p^{(4)} \in P_4$ is a heavy $4$-tuple and $p^{(3)} \in P_3$ is a heavy triple, and hyperedges in $H^{(4,2)}_1, \dots, H^{(4,2)}_k$ have the form $(p^{(4)}, p^{(2)}, v)$, where $p^{(4)} \in P_4$, $p^{(2)} \in P_2$, and $v \in [n]$. (Because $\cup_{i = 1}^k H_i$ is the union of matchings, each variable $v$ appears in at most $k$ hyperedges, so it is not possible to have a heavy singleton.)

For this overview, let us consider the case of $H^{(4,2)}_1, \dots, H^{(4,2)}_k$. We need to bound the value of $\Phi^{(4,2)}_b(x,y)$, defined as
\begin{equation*}
\Phi^{(4,2)}_b(x,y) \defeq \sum_{i = 1}^k b_i \sum_{(p^{(4)}, p^{(2)}, v) \in H^{(4,2)}_i} y_{p^{(4)}} y_{p^{(2)}} x_v \mper
\end{equation*} 
As before, we can now derive constraints using the Cauchy--Schwarz trick. Namely, we can take two hyperedges $(p^{(4)}, p^{(2)}_1, v_1) \in H^{(4,2)}_i$ and $(p^{(4)}, p^{(2)}_2, v_2) \in H^{(4,2)}_j$ that share the same heavy $4$-tuple $p^{(4)}$, and then form the derived hyperedge $((p^{(2)}_1, v_1), (p^{(2)}_2, v_2))$. 

Now, the approach of~\cite{AlrabiahGKM23} breaks down. To use their Kikuchi graph, we need to be able to derive constraints that only use the original variables $[n]$. But, the above derived constraints still use ``heavy pair'' variables. One could try to, e.g., combine the derived constraint $((p^{(2)}_1, v_1), (p^{(2)}_2, v_2))$ with some constraint in $H^{(4,2)}_{j'}$ that also contains the heavy pair $p^{(2)}_1$, but such a constraint will be of the form $(p^{(4)}_2, p^{(2)}_2, v_3)$, i.e., it will have a new heavy $4$-tuple. So, the new derived constraint will be $((p^{(2)}_1, v_1), v_2, (p^{(4)}_2, v_3))$, and we have the same problem again. Furthermore, we cannot try to combine different ``types'' of hypergraphs, e.g., combine constraints in $H^{(2)}_j$ for some $j$ with constraints in $H^{(4,2)}_i$ for some $i$, as it could be the case that after the hypergraph decomposition step, \emph{most} (or all) of the original $\delta n k$ hyperedges are placed in, e.g., $H^{(4,2)}_i$ for some $i$, and so all hypergraphs of a different ``type'' are empty.

Let us now explain our approach to handle this problem. We need to design a Kikuchi matrix for hyperedges that are \emph{partite}: each hyperedge in $\cup_{i = 1}^k H^{(4,2)}_i$ has two vertices from the vertex set $P_2$ and $2$ vertices from the vertex set $[n]$. We introduce the following Kikuchi graph in this work for partite hypergraphs, which is defined as follows. For a derived constraint $((p^{(2)}_1, v_1), (p^{(2)}_2, v_2))$, we let the matrix $A_{((p^{(2)}_1, v_1), (p^{(2)}_2, v_2))}$ be the matrix indexed by \emph{pairs} of sets $S_1$ and $S_2$, where $S_1 \subseteq [n]$ has size $\ell$ and $S_2 \subseteq P_2$ also has size $\ell$. That is, each row has a set for each distinct set of variables, which are $[n]$ and $P_2$. By analogy to the earlier definition of $A_C$, we should set $A_{((p^{(2)}_1, v_1), (p^{(2)}_2, v_2))}((S_1, S_2), (T_1, T_2)) = 1$ if $S_1 \oplus T_1 = \{v, v'\}$ and $S_2 \oplus T_2 = \{p^{(2)}_1, p^{(2)}_2\}$, and furthermore we require that the variables ``coming from'' the underlying original hyperedges in $H^{(4,2)}$ are split \emph{as evenly as possible} across the rows and columns. Namely, we require that either $v_1 \in S_1, p^{(2)}_2 \in S_2$ and $v_2 \in T_1, p^{(2)}_1 \in T_2$, or vice versa, so that the variables $(v_1, p^{(2)}_1)$ ``coming from'' the first underlying hyperedge are split across the row and column, and likewise for the second underlying hyperedge.

Analogously to the definitions in \cref{sec:agkmodd}, we can randomly partition $[k]$ into $L \cup R$ and let $A_{i,j} \defeq \sum_{p^{(4)} \in P_4} \sum_{\substack{(p^{(4)}, p^{(2)}_1, v_1) \in H^{(4,2)}_i \\ (p^{(4)}, p^{(2)}_2, v_2) \in H^{(4,2)}_j}} A_{((p^{(2)}_1, v_1), (p^{(2)}_2, v_2))}$, $A_i \defeq \sum_{j \in R} b_j A_{i,j}$ and $A = \sum_{i \in L} b_i A_i$. A straightforward calculation shows that $\E_b[\norm{A}_2]$ provides an upper bound on $\val(\Phi^{(4,2)}_b)$. Thus, to determine whether or not this matrix $A$ is good enough to prove the desired $n \geq \tilde{\Omega}(k^{\frac{q}{q-2}})$ bound, we need to argue that the row pruning step holds, i.e., each $A_i$ can be made approximately regular.

It turns out that, while the calculations are substantially more complicated that those appearing in \cite{AlrabiahGKM23}, this approach does in fact work, provided that we adjust the thresholds $d_s$ slightly. We need to set $d_s = n^{-\frac{2s}{q} + \frac{2}{q}} k$ (instead of $n^{-\frac{2s}{q} + \frac{3}{q}} k$) for $2 \leq s \leq \frac{q-1}{2}$, while keeping $d_{\frac{q+1}{2}} = n^{\frac{2}{q} - 1} k$ the same. This allows us to prove the $n \geq \tilde{\Omega}(k^{\frac{q}{q-2}})$ bound for $q = 7$.

\parhead{Generalizing our $q = 7$ approach to all odd $q$.} We can now generalize our above approach to all odd $q$ as follows. The output of the recursive hypergraph decomposition step yields hypergraphs $H^{(s, t_1, \dots, t_r)}_1, \dots, H^{(s, t_1, \dots, t_r)}_k$, where $2 \leq s \leq \frac{q+1}{2}$ is an integer, $s \geq t_1 \geq \dots \geq t_r \geq 2$ are all positive integers, and $t_1 + t_2 + \dots + t_r + s \leq q$. This notation means that each hypergraph $H^{(s, t_1, \dots, t_r)}_i$ contains hyperedges of the form $(p^{(s)}, p^{(t_1)}, \dots, p^{(t_r)}, C)$, where each $p^{(t_z)} \in P_{t_z}$, i.e., it is a heavy $t_z$-tuple, $p^{(s)} \in P_s$ is a heavy $s$-tuple, and $C \subseteq [n]$ has size $q - s - t_1 - \dots - t_r$ and is the set of remaining ``original variables'' from $[n]$.
We then apply the Cauchy--Schwarz trick to form derived constraints, which now take the form $((p^{(t_1)}_1, \dots, p^{(t_r)}_1, C_1), (p^{(t_1)}_2, \dots, p^{(t_r)}_2, C_2))$. We define the Kikuchi matrix $A_{((p^{(t_1)}_1, \dots, p^{(t_r)}_1, C_1), (p^{(t_1)}_2, \dots, p^{(t_r)}_2, C_2))}$ analogously to the case of $H^{(4,2)}$ shown above for $q= 7$. Namely, we make the following definition.
\begin{definition}[Kikuchi matrices for partite hypergraphs]
\label{def:partitekikuchi}
Let $((p^{(t_1)}_1, \dots, p^{(t_r)}_1, C_1), (p^{(t_1)}_2, \dots, p^{(t_r)}_2, C_2))$ be a derived constraint where $p^{(t_z)}_1, p^{(t_z)}_2 \in P_{t_z}$ for each $z \in [r]$ and $C_1, C_2 \subseteq [n]$ have size $q - s - \sum_{z = 1}^r t_z$. Let $A_{((p^{(t_1)}_1, \dots, p^{(t_r)}_1, C_1), (p^{(t_1)}_2, \dots, p^{(t_r)}_2, C_2))}$ be the matrix indexed by tuples of sets $(S_1, \dots, S_r, S)$ where for $z \in [r]$, $S_z \subseteq P_z$ has size $\ell$ and $S \subseteq [n]$ has size $\ell$. The matrix has a $1$ in the $((S_1, \dots, S_r, S), (T_1, \dots, T_r, T))$-th entry if for each, $S_z \oplus T_z = \{p^{(t_z)}_1, p^{(t_z)}_2\}$, and $S \oplus T = C_1 \oplus C_2$. Otherwise, the entry is $0$.
\end{definition}
In order for the row pruning step to go through, we need to be a bit more careful in our definition of the matrix $A_{((p^{(t_1)}_1, \dots, p^{(t_r)}_1, C_1), (p^{(t_1)}_2, \dots, p^{(t_r)}_2, C_2))}$. Namely, as done in \cref{sec:agkmodd,sec:recursivedecomp}, we want to split the variables $(p^{(t_1)}_1, \dots, p^{(t_r)}_1, C_1)$ and $(p^{(t_1)}_2, \dots, p^{(t_r)}_2, C_2)$ of the underlying hyperedges \emph{evenly} across the row and column sets $(S_1, \dots, S_r, S)$ and $(T_1, \dots, T_r, T)$. However, each element $p^{(t_z)}$ comes with a ``weight'' of $t_z$ because $p^{(t_z)}$ corresponds to a set of $t_z$ original variables $[n]$. Because we ``canceled out'' a variable in $P_s$, the total weight of the remaining elements in each hyperedge is $q - s$. Ideally, we would like to achieve an even split of $(\frac{q-s}{2}, \frac{q-s}{2})$, but this is not always possible.

Fortunately, for the row pruning step to go through, we do not require an exactly even split: we just need that each side of the split has total weight at most $\frac{q}{2}$. A simple greedy algorithm shows that this is in fact always possible, and so we can make the row pruning step go through. This allows us to prove a lower bound of $n \geq \tilde{\Omega}(k^{\frac{q}{q-2}})$ bound for all odd $q$, finishing the proof.
}

\subsection{A proof from bipartite Kikuchi graphs}
\label{sec:bipartitekikuchiintro}
\ignore[old preamble]{
The above proof for odd $q$ is substantially more complicated compared to the fairly simple proof for even $q$ sketched in~\cref{sec:agkmeven}. Recall that this is the case because the proof in~\cref{sec:agkmodd} uses the ``Cauchy--Schwarz trick'' to derive new constraints of arity $2(q-1)$ that have \emph{correlated} randomness. When we group the derived constraints so that they have independent randomness, i.e., we write $A = \sum_{i \in L} b_i A_i$, where each $b_i$ is an independent bit, the success of the row pruning step is dictated by the structure of the derived constraints that contribute to the matrix $A_i$. However, the derived constraints look like $(C, C')$ where there exists $u \in [n]$ such that $(u, C) \in H_i$ and $(u,C') \in \cup_{j \in R} H_j$, and so not only are they no longer matchings, they can have highly irregular structure. Because of this, we then decompose the original hypergraph matchings $H_1, \dots, H_k$ into new hypergraphs where the union, over $i \in [k]$ of the new hypergraphs in each ``piece'' of the decomposition is regular. In the case of $q = 3, 5$, this can be handled with a simple decomposition step as done (or could have been done for $q = 5$) in~\cite{AlrabiahGKM23}, and for $q \geq 7$ a more involved recursive hypergraph decomposition step along with a partite Kikuchi matrix, as sketched in \cref{sec:recursivedecomp,def:partitekikuchi}, is needed.

The key point here is that original hypergraphs $H_1, \dots, H_k$ are matchings, i.e., they are already ``regular'', and so the additional complexity in the proof for odd $q$ comes from the use of the Cauchy--Schwarz trick, which is the step that derives new constraints that are no longer ``regular''. If we could somehow avoid using the Cauchy--Schwarz trick entirely, we would never lose the ``regularity property'', and so we could potentially obtain a substantially simpler and more direct proof that avoids any hypergraph decomposition steps.
}

As explained in \cref{sec:agkmodd}, the approach of~\cite{AlrabiahGKM23} uses the ``Cauchy--Schwarz trick'' to derive new constraints of arity $2(q-1)$ that have \emph{correlated} randomness. When we group the derived constraints so that they have independent randomness, i.e., we write $A = \sum_{i \in L} b_i A_i$, where each $b_i$ is an independent bit, the success of the row pruning step is dictated by the structure of the derived constraints that contribute to the matrix $A_i$. However, the derived constraints are $(C, C')$ where there exists $u \in [n]$ such that $(u, C) \in H_i$ and $(u,C') \in \cup_{j \in R} H_j$. That is, not only are the derived constraints no longer matchings, they can have highly irregular structure, and so the row pruning step fails. The hypergraph decomposition step is an attempt to form derived constraints so that the derived constraints are regular, but as we have explained, this step breaks for $q \geq 7$.

In contrast, in the case of even $q$ (\cref{sec:agkmeven}), the constraints that contribute to the matrix $A_i$ come from the hypergraph matching $H_i$, i.e., they are already ``regular''. Thus, the additional complexity in the proof for odd $q$ comes from the use of the Cauchy--Schwarz trick, which is the step that derives new constraints that are no longer ``regular''. If we could somehow avoid using the Cauchy--Schwarz trick entirely, we would never lose the ``regularity property'', and so there would be no need for a hypergraph decomposition step.  As the hypergraph decomposition step is the place where the approach in~\cref{sec:agkmodd} fails, this would be a possible route to prove~\cref{mthm:oddqldc}.

Unfortunately, we are not quite able to achieve this goal --- in the ``top level'' case, i.e., when the hypergraph matchings $H_1, \dots, H_k$ already satisfy the global property that $\cup_{i \in [k]} H_i$ has no heavy sets of size $s$ for $2 \leq s \leq \frac{q+1}{2}$, we still do the Cauchy--Schwarz trick to construct the Kikuchi matrix as sketched in~\cref{sec:agkmodd}. This means that we will do \emph{one} step of hypergraph decomposition to produce, for each $2 \leq s \leq \frac{q+1}{2}$, a decomposed instance with hypergraph matchings $H^{(s)}_1, \dots, H^{(s)}_k$.

Now, recall that the approach in~\cref{sec:agkmodd} fails for the following reason: if we were to apply the Cauchy--Schwarz trick again to each decomposed instance $H^{(s)}_1, \dots, H^{(s)}_k$, the resulting derived constraints may again not be regular. However, the hypergraphs $H^{(s)}_1, \dots, H^{(s)}_k$ are themselves still matchings, i.e., they are regular, as this property is inherited from the original hypergraphs $H_1, \dots, H_k$. Our key technical contribution, as we now explain, is the introduction of a \emph{bipartite} Kikuchi graph that allows us to refute each decomposed instance $H^{(s)}_1, \dots, H^{(s)}_k$ \emph{without} using the Cauchy--Schwarz trick at all. Because we do not apply the Cauchy--Schwarz trick, our hypergraphs remain matchings, and so we avoid the proof barrier sketched in~\cref{sec:agkmodd}. The thresholds that we use in the decomposition step are the thresholds $d_s = n^{-\frac{2s}{q} + \frac{3}{q}} k$, the same as in~\cref{sec:agkmodd}.

\parhead{Refuting the decomposed instances with bipartite Kikuchi graphs.} Instead of doing the Cauchy--Schwarz trick, we introduce a bipartite Kikuchi graph that is imbalanced. This is perhaps a counterintuitive approach to try, as typically imbalanced matrices do not give good spectral refutations. For example, as explained at the beginning of~\cref{sec:agkmodd}, one can define an imbalanced matrix (\cref{def:naiveoddkikuchi}) that cleanly handles the case of odd $q$ with no hypergraph decomposition steps at all, but the imbalance of the matrix produces the weaker bound of $n \geq \tilde{\Omega}(k^\frac{q+1}{q-1})$. Intuitively, the reason that the imbalanced matrix of \cref{def:naiveoddkikuchi} fails to prove a bound of $n \geq \tilde{\Omega}(k^\frac{q}{q-2})$ is that the average right degree in the matrix is roughly $n^{-1/q} = o(1)$ when we set $\ell$ to be the desired value of $n^{1 - 2/q}$, whereas in our new bipartite graph that we define shortly, both the left and right degrees will be $\Omega(1)$.

Recall that for each $i \in [k]$, a hyperedge in $H^{(s)}_i$ has the form $(C,p)$ where $p \in P_s$ and $\abs{C} = q - s$. We now define our bipartite Kikuchi graph $A_{C,p}$.
\begin{definition}[Our imbalanced bipartite Kikuchi graph]
\label{def:bipartitekikuchintro}
For a set $C \in {[n] \choose q - s}$ and $p \in P_s$, let $A_{C,p}$ be the adjacency matrix of the following graph. The left vertices are pairs of sets $(S_1, S_2)$ where $S_1 \subseteq [n]$ has size $\ell$ and $S_2 \subseteq P_s$ has size $\ell$ as well. The right vertices are pairs of sets $T_1 \subseteq [n]$ and $T_2 \subseteq P_s$, where $\abs{T_1} = \ell + 1 - s$ and $\abs{T_2} = \ell + 1$. We put an edge $((S_1, S_2), (T_1, T_2))$ if $S_1 \oplus T_1 = C$, which implies that $\abs{S_1 \cap C} = \frac{q - 1}{2}$ and $\abs{T_1 \cap C} = \frac{q+1}{2} - s$, and also $S_2 \oplus T_2 = \{p\}$, which implies that $p \notin S_2$ and $p \in T_2$.
\end{definition}
The idea behind \cref{def:bipartitekikuchintro} is that for each row/column, we have a subset for each ``variable set'', i.e., $[n]$ and $P_s$. However, in order for $A_{C,p}$ to have nonzero entries, we must have $\abs{S_1} \ne \abs{T_1}$ and $\abs{S_2} \ne \abs{T_2}$, and this makes the graph $A_{C,p}$ quite imbalanced. The size of the left vertex set is $N_L = {n \choose \ell} {\abs{P_s} \choose \ell}$, where $\abs{P_s} \leq n k /d_s$. This bound on $\abs{P_s}$ follows because $\cup_{i \in [k]} H_i$ has at most $nk$ hyperedges in total and each $p \in P_s$ is contained in at least $d_s$ hyperedges. On the other hand, the size of the right vertex set is $N_R = {n \choose \ell + 1 - s} {\abs{P_s} \choose \ell + 1} \leq \frac{\abs{P_s}}{\ell} \left(\frac{\ell}{n}\right)^{s - 1} N_L$. Recall that $\ell = n^{1 - \frac{2}{q}}$ and $d_s =  n^{-\frac{2s}{q} + \frac{3}{q}} k$, so $N_R \leq n^{1/q} N_L$. Moreover, if $\abs{H_i^{(s)}} \approx \abs{H_i}$ for most $i$, i.e., most hyperedges contain a heavy set $Q$ of size $s$, and each heavy $Q$ of size $s$ appears in $O(d_s)$ hyperedges, then the above inequalities also approximately hold in the reverse, so that $N_R \approx n^{1/q} N_L$. For this proof overview, we will assume for simplicity both of these assumptions hold.

\begin{remark}
\label{rem:bipartitekikuchi}
We note that $A_{C,p}$ is not the first use of an imbalanced bipartite Kikuchi graph, as imbalanced Kikuchi graphs are used in~\cite{Yankovitz24, KothariM24}. However, in those works one can easily produce an equivalent Kikuchi graph (i.e., non-bipartite and balanced) with essentially the same properties via one application of the Cauchy--Schwarz derivation trick to the underlying hyperedges. However, as discussed in detail in~\cref{sec:agkmodd}, such a strategy fails for $q \geq 7$. We thus view our bipartite Kikuchi graph $A_{C,p}$ as being \emph{inherently} bipartite, as we are unable to construct an equivalent (balanced) Kikuchi graph with analogous properties to it by first deriving constraints on the underlying hypergraph and then forming a balanced Kikuchi graph similar to \cref{def:basickikuchi} using the derived constraints. We admit that the notion of being ``inherently bipartite'' is informal, but we do not see an easy way to formalize this.
\end{remark}

With the graph $A_{C,p}$ defined, we then let $A_i = \sum_{(C,p) \in H^{(s)}_i} A_{C,p}$ and $A = \sum_{i = 1}^k b_i A_i$. Crucially, because we have not used the Cauchy--Schwarz trick, $A_i$ depends only  on $H^{(s)}_i$ and not on any of the other hypergraphs. Because $H^{(s)}_i$ is a matching (over the larger vertex set $[n]$ and $P_s$), the row pruning calculation is much more straightforward.

\parhead{Row pruning for the imbalanced bipartite Kikuchi graph.} Unlike the case of normal Kikuchi graphs, we now need to argue concentration of both the left and right vertices. The case of the left vertices is rather straightforward: because $H^{(s)}_i$ is a matching and for an edge in $A_{C,p}$, the left vertex contains $\frac{q-1}{2}$ elements of $C$, which is the subset of ``original variables'' $[n]$, a similar calculation to the row pruning argument in~\cref{sec:agkmeven} shows concentration of the degrees of the left vertices provided that the average left degree $d_L$ is at least $\Omega(1)$. The average left degree is $d_L \approx \left(\frac{\ell}{n}\right)^{\frac{q-1}{2}} n$, which is roughly $n^{1/q} \gg 1$ since $\ell = n^{1 - 2/q}$.

The calculation for the right vertices is more interesting. We again use that $H^{(s)}_i$ is a matching to argue concentration provided that the average right degree $d_R$ is at least $\Omega(1)$. Now, we have $d_R \approx \left(\frac{\ell}{n}\right)^{\frac{q+1}{2} - s}  \left(\frac{\ell}{\abs{P_s}}\right) \cdot n$. As shown earlier, $\abs{P_s} \leq nk/d_s \approx n^{1+\frac{2s}{q} - \frac{3}{q}}$, using our threshold for $d_s$. Substituting in $\ell = n^{1 - 2/q}$ and the above bound on $\abs{P_s}$, we see that $d_R \geq \Omega(1)$ holds. Thus, the row pruning step goes through, and we are able to prove the $n \geq \tilde{\Omega}(k^{\frac{q}{q-2}})$ bound.

\parhead{Our bipartite Kikuchi graph compared to the naive imbalanced matrix.}
Why does the matrix $A_{C,p}$ succeed in yielding a $k^{\frac{q}{q-2}}$ lower bound, whereas the naive $\ell$ vs.\ $\ell+1$ matrix (\cref{def:naiveoddkikuchi}) only yields a $k^{\frac{q+1}{q-1}}$ lower bound?  Recall that for an edge $((S_1, S_2), (T_1, T_2))$ in the matrix $A_{C,p}$, we have split $C$ across $S_1$ and $T_1$ as $\abs{S_1 \cap C} = \frac{q-1}{2}$ and $\abs{T_1 \cap C} = \frac{q+1}{2} - s$. Because $p \notin S_2$ and $p \in T_2$, this means that the row set contains $\frac{q-1}{2}$ variables from $(C,p)$, and the column set ``effectively'' contains $\frac{q+1}{2}$ variables from $(C,p)$: it has $\frac{q+1}{2} - s$ variables contained in $T_1$ and then an extra $s$ from $T_2$ because $p_s$ as a variable ``represents'' a set of size $s$. Notice that the $\frac{q-1}{2}$ vs.\ $\frac{q+1}{2}$ split is precisely the split used by the $\ell$ vs.\ $\ell+1$ matrix. So, it is reasonable to ask: why is the matrix $A_{C,p}$ performing better?

The reason lies in the fact that the chosen threshold is $d_s =  \left(\frac{\ell}{n}\right)^{s - \frac{1}{2}} k = n^{-\frac{2s}{q} + \frac{3}{q}} k$ has an ``extra factor'' of $n^{1/q}$ when viewed from the following perspective. Recall that the thresholds $d_s$ are used only for $2 \leq s \leq \frac{q+1}{2}$. But, if we substitute in $s = 1$, we get $d_1 = n^{1/q} k$, which we can view as giving us a bound on the maximum degree of a singleton $v$ that we are able to tolerate. However, in the hypergraph $H \defeq \cup_{i = 1}^k H_i$, each singleton $v$ has $\deg_H(v) \leq k$, as $H$ is a union of matchings. So, $d_s$ is a $n^{1/q}$ factor ``larger'' than we might expect. In fact, following intuition from~\cite{GuruswamiKM22}, we would like $\deg(Q)$ to ``drop'' by a factor of $\ell/n$ per additional vertex included into the set $Q$, i.e., we intuitively set $d'_1 = k$ and then take $d'_{s+1} = (\ell/n) d'_s$ for larger $s$. This yields the threshold $d'_s = \left(\frac{\ell}{n}\right)^{s - 1} k = n^{-\frac{2s}{q} + \frac{2}{q}} k$. However, the threshold $d'_s$ is \emph{not} the threshold that arises out of the approach of~\cite{AlrabiahGKM23} (\cref{sec:agkmodd}); we can tolerate an extra factor of $n^{1/q}$.

Let us now explain why this $n^{1/q}$ factor is critical.
We expect the left/right degrees to be about $\left(\frac{\ell}{n}\right)^{\text{\# variables in split}} \cdot n$, i.e., $\left(\frac{\ell}{n}\right)^{\frac{q-1}{2}} \cdot n$ for the left degree and $\left(\frac{\ell}{n}\right)^{\frac{q+1}{2}} \cdot n$ for the right degree. If this happens, then we must take $\ell = n^{1 - 2/(q+1)}$ rather than $n^{1 - 2/q}$. This results in the weaker $k^{\frac{q+1}{q-1}}$ bound, and is precisely what happens with the naive matrix from \cref{def:naiveoddkikuchi}. However, for the matrix $A_{C,p}$, the fact that $d_s$ has an extra factor of $n^{1/q}$ boosts the average right degree by a factor of $n^{1/q}$. This means that our right degree is roughly $\left(\frac{\ell}{n}\right)^{\frac{q+1}{2}} \cdot n^{1 + 1/q}$, and this allows us to still take $\ell = n^{1 - 2/q}$. Notice that when we computed the sizes of the left and right vertex sets  for $A_{C,p}$, we had $N_R \approx n^{1/q} N_L$, whereas in the naive imbalanced matrix of \cref{def:naiveoddkikuchi} one has $N_R \approx n^{2/q} N_L$.

\parhead{Postmortem: the power of the bipartite Kikuchi matrix.} The proof of the $n \geq \tilde{\Omega}(k^{\frac{q}{q-2}})$ lower bound using the bipartite Kikuchi matrices sketched above is substantially simpler than a proof one might expect to obtain by following the more well-trodden path of ``hypergraph decomposition + Cauchy--Schwarz'' used in prior works (\cite{GuruswamiKM22,HsiehKM23,AlrabiahGKM23,KothariM23,Yankovitz24,KothariM24}). The success of the bipartite matrix in this setting comes as quite a surprise to us, as it is contrary to the conventional wisdom that imbalanced matrices yield poorer spectral certificates compared to balanced matrices. Moreover, the simplicity of the analysis is not just nice for aesthetic reasons: \cref{mthm:oddqldc} obtains a better dependence on $\log k$, $\delta$, and $\eps$ in the case of $q = 3$ as compared to the lower bound of~\cite{AlrabiahGKM23}. In fact, the $\log k$ dependence in \cref{mthm:oddqldc} is exactly the same as the dependence obtained for even $q$, and the loss in the $\delta$ and $\eps$ factors comes from the ``top level'' instance where we still use the Cauchy--Schwarz trick.

\parhead{Roadmap.} The full proof of \cref{mthm:oddqldc} is presented in \cref{sec:setup,sec:regularref,sec:bipartiteref}; we give preliminaries and notation in \cref{sec:prelims}. In \cref{sec:setup}, we handle the setup and the hypergraph decomposition step (\cref{sec:decomp}). In \cref{sec:regularref}, we refute the ``top level'' instance using the Cauchy--Schwarz trick, and in \cref{sec:bipartiteref} we use our new bipartite Kikuchi matrices to refute the subinstances $H^{(s)}_1, \dots, H^{(s)}_k$ for all $2 \leq s \leq \frac{q+1}{2}$.

\section{Preliminaries}
\label{sec:prelims}

\subsection{Basic notation and hypergraphs}
We let $[n]$ denote the set $\{1, \dots, n\}$. For two subsets $S, T \subseteq [n]$, we let $S \oplus T$ denote the symmetric difference of $S$ and $T$, i.e., $S \oplus T \coloneqq \{i : (i \in S \wedge i \notin T) \vee (i \notin S \wedge i \in T)\}$. For a natural number $t \in \N$, we let ${[n] \choose t}$ be the collection of subsets of $[n]$ of size exactly $t$.

For a rectangular matrix $A \in \R^{m \times n}$, we let $\norm{A}_2 \coloneqq \max_{x \in \R^m, y \in \R^n: \norm{x}_2 = \norm{y}_2 = 1} x^{\top} A y$ denote the spectral norm of $A$.

\begin{definition}[Hypergraphs and hypergraph matchings]
\label{def:hypergraphs}
A hypergraph $H$ with vertices $[n]$ is a collection of subsets $C \subseteq [n]$ called hyperedges. We say that a hypergraph $H$ is \emph{$q$-uniform} if $\abs{C}= q$ for all $C \in H$, and we say that $H$ is a \emph{matching} if all the hyperedges in $H$ are disjoint. For a subset $Q \subseteq [n]$, we define the degree of $Q$ in $H$, denoted $\deg_{H}(Q)$, to be $\abs{\{C \in H : Q \subseteq C\}}$.
\end{definition}

\begin{definition}[Bipartite hypergraphs]
\label{def:bipartitehypergraphs}
A bipartite hypergraph $H$ has two vertex sets $[n]$ and $P$ and is a collection of pairs $(C,p)$ with $C \subseteq [n]$ and $p \in P$ called hyperedges. We say that a bipartite hypergraph $H$ is \emph{$q$-uniform} if $\abs{C}= q - 1$ for all $(C,p) \in H$, and we say that $H$ is a \emph{matching} if all the hyperedges in $H$ are disjoint. That is, for $(C,p)$ and $(C',p')$ in $H$, it holds that $C \cap C' = \emptyset$ and $p \ne p'$.
\end{definition}

\subsection{Locally decodable codes}
We refer the reader to the survey~\cite{Yekhanin12} for background. 

A code is typically defined as a map $\Code \colon \Bits^k \to \Bits^n$. However, for our proofs it will be more convenient to view a code as taking values in $\Fits$ rather than $\Bits$; we switch between the two notations via the map $0 \mapsto 1$ and $1 \mapsto -1$. For a code $\Code \colon \Fits^k \to \Fits^n$, we will write $x \in \Code$ to denote an $x = \Code(b)$ for some $b \in \Bits^k$.

A locally decodable code is a code where one can recover any bit $b_i$ of the original message $b$ with good confidence while only reading a few bits of the encoded string in the presence of errors.
\begin{definition}[Locally decodable code]
\label{def:LDC}
A code $\Code \colon \Fits^k \to \Fits^n$ is $(q, \delta, \eps)$-locally decodable if there exists a randomized decoding algorithm $\Dec(\cdot)$ with the following properties. The algorithm $\Dec(\cdot)$ is given oracle access to some $y \in \Fits^n$, takes an $i \in [k]$ as input, and satisfies the following: \begin{inparaenum}[(1)] \item the algorithm $\Dec$ makes at most $q$ queries to the string $y$, and \item for all $b \in \Fits^k$, $i \in [k]$, and all $y \in \Fits^n$ such that $\Delta(y, \Code(b)) \leq \delta n$, $\Pr[\Dec^{y}(i) = b_i] \geq \frac{1}{2} + \eps$. Here, $\Delta(x,y)$ denotes the Hamming distance between $x$ and $y$, i.e., the number of indices $v \in [n]$ where $x_v \ne y_v$.\end{inparaenum}
\end{definition}

Following known reductions \cite{Yekhanin12}, locally decodable codes can be reduced to the following normal form, which is more convenient to work with.
\begin{definition}[Normal LDC]
\label{def:normalLDC}
A code $\Code \colon \Fits^k \to \Fits^n$ is $(q, \delta, \eps)$-normally decodable if for each $i \in [k]$, there is a $q$-uniform hypergraph matching $H_i$ with at least $\delta n$ hyperedges such that for every $C \in H_i$, it holds that $\Pr_{b \gets \Fits^k}[b_i = \prod_{v \in C} \Code(b)_v] \geq \frac{1}{2} + \eps$.
\end{definition}

\begin{fact}[Reduction to LDC normal form, Lemma 6.2 in~\cite{Yekhanin12}]\label{fact:LDCnormalform}
Let $\Code \colon \Fits^k \to \Fits^n$ be a code that is $(q, \delta, \eps)$-locally decodable. Then, there is a code $\Code' \colon \Fits^k \to \Fits^{O(n)}$ that is  $(q, \delta', \eps')$ normally decodable, with $\delta' \geq \eps \delta/3q^2 2^{q-1}$ and $\eps' \geq \eps/2^{2q}$.
\end{fact}

\subsection{Matrix concentration inequalities}
We will make use of the following non-commutative Khintchine inequality~\cite{LustPiquardG91}.
\begin{fact}[Rectangular Matrix Khintchine Inequality, Theorem 4.1.1 of \cite{Tropp15}]
\label{fact:matrixkhintchine}
Let $X_1, \dots, X_k$ be fixed $d_1 \times d_2$ matrices and $b_1, \dots , b_k$ be i.i.d.\ from $\Fits$. Let $\sigma^2 \geq \max(\norm{\sum_{i = 1}^k X_i X_i^{\top}]}_2, \norm{\sum_{i = 1}^k X_i^{\top} X_i]}_2)$. Then
\begin{equation*}
\E\Bigl[\ \norm{\sum_{i = 1}^k b_i X_i}_2\ \Bigr] \leq \sqrt{2\sigma^2 \log(d_1 + d_2)} \mper
\end{equation*}
\end{fact}

\subsection{Binomial coefficient inequalities}
In this section, we state and prove the following fact about binomial coefficients that we will use.
\begin{fact}
\label{fact:binomest}
Let $n, \ell, q$ be positive integers with $\ell \leq n$. Let $q$ be constant and $\ell, n$ be asymptotically large with $\ell = o(n)$. Then, 
\begin{flalign*}
&\frac{ {n \choose \ell - q}}{ {n \choose \ell }} = \Theta\left(\left(\frac{\ell}{n}\right)^q\right) \mcom \\
&\frac{ {n - q \choose \ell} }{{n \choose \ell}} = \Theta(1) \mper
\end{flalign*}
\end{fact}
\begin{proof}
We have that
\begin{flalign*}
&\frac{ {n \choose \ell - q}}{ {n \choose \ell }} = \frac{ { \ell \choose q}}{ {n - \ell + q \choose q}} \mper
\end{flalign*}
Using that $\left(\frac{a}{b} \right)^{b} \leq {a \choose b} \leq \left(\frac{e a}{b} \right)^{b}$ finishes the proof of the first equation.

We also have that
\begin{flalign*}
&\frac{ {n - q \choose \ell} }{{n \choose \ell}} = \frac{(n - q)! (n - \ell)!}{n! (n - \ell - q)!} = \prod_{i = 0}^{q -1 } \frac{n - \ell - i}{n - i} = \prod_{i = 0}^{q-1} \left(1 - \frac{\ell}{n - i}\right) \mcom
\end{flalign*}
and this is $\Theta(1)$ since $\ell = o(n)$ and $q$ is constant.
\end{proof}

\section{Proof of \cref{mthm:oddqldc}}
\label{sec:setup}
In this section, we begin the proof of \cref{mthm:oddqldc}. By \cref{fact:LDCnormalform}, we may assume that we start with a code $\Code$ in normal form. Namely, $\Code$ is map $\Code \colon \Fits^k \to \Fits^n$, and there exist $q$-uniform hypergraphs $H_1, \dots, H_k$ of size exactly $\delta n$ such that for every $C \in H_i$, it holds that $\E_{b \gets \Fits^k}[b_i \prod_{v \in C} \Code(b)_v] \geq \eps$. We will show that $k \leq O(n^{1 - \frac{2}{q}} \delta^{-2 - \frac{2}{q}} \eps^{-4})$ when $\Code$ is in normal form (\cref{def:normalLDC}), which implies \cref{mthm:oddqldc} for standard LDCs (\cref{def:LDC}).

To begin, we let $\Phi_b(x)$ denote the following polynomial:
\begin{equation*}
\Phi_b(x) \defeq \sum_{i = 1}^k \sum_{C \in H_i} b_i \prod_{v \in C} x_v \mper
\end{equation*}
Because $\E_{b \gets \Fits^k}[b_i \prod_{v \in C} \Code(b)_v] \geq \eps$, it follows that $\E_b[\Phi_b(\Code_b(x))] \geq \eps \sum_{i = 1}^k \abs{H_i} \geq \eps \delta nk$. Hence, we have that $\E_b[\val(\Phi_b)] \geq \eps \delta n k$, where $\val(\Phi_b) \defeq \max_{x \in \Fits^n} \Phi_b(x)$.

\parhead{Overview: refuting the $q$-XOR instance $\Phi_b$.} It thus remains to bound $\E_b[\val(\Phi_b)]$. We will do this by building on the spectral methods of~\cite{GuruswamiKM22,AlrabiahGKM23}. As discussed in \cref{sec:bipartitekikuchiintro}, the argument proceeds in two steps.
\begin{enumerate}[(1)]
\item \textbf{Hypergraph decomposition:} First, we decompose the hypergraphs $H_1, \dots, H_k$ informally as follows. For $2 \leq t \leq \frac{q + 1}{2}$, we define ``degree thresholds'' $d_t$ where $d_t \coloneqq \left(\frac{\ell}{n}\right)^{t - \frac{3}{2}} k$. 
Then, for every $Q \subseteq [n]$ of size $s$ with $2 \leq s \leq \frac{q + 1}{2}$, we call $Q$ ``heavy'' if $Q$ is contained in more than $d_s$ hyperedges in the multiset $\cup_{i  = 1}^k H_i$. For each heavy $Q$, we introduce a new variable $y_{p_Q}$ and let $P_s$ be the set of the labels $p_Q$ corresponding to $\abs{Q} = s$. Then, for each hyperedge $C$, if $Q \subseteq C$ is the largest heavy $Q$ contained in $C$, we replace the hyperedge $C$ with $(C \setminus Q, p_Q)$, where $p_Q \in P_s$. We thus produce, for each $i \in [k]$, hypergraphs $H_i^{(s)}$ for each $2 \leq s \leq \frac{q+1}{2}$ where a hyperedge in $H_i^{(s)}$ has the form $(C', p)$ for some $p \in P_s$ and $C' \subseteq [n]$ with $\abs{C'} = q - s$, along with the hypergraph $H'_i$ of ``leftover edges''.
\item \textbf{Refutation:} With the decomposition in hand, we then produce polynomials $\Psi_b^{(s)}$ for each $2 \leq s \leq \frac{q+1}{2}$, along with a polynomial $\Psi_b$, such that $\Psi_b + \sum_{s = 2}^{\frac{q+1}{2}} \Psi_b^{(s)} = \Phi_b$. We then produce upper bound $\E_b[\val(\Psi_b)]$ as well as $\E_b[\val(\Psi_b^{(s)})]$ for each polynomial $\Psi_b^{(s)}$ in the decomposition. Combining these bounds allows us to upper bound $\E_b[\val(\Phi_b)]$ and finishes the proof. As discussed in \cref{sec:bipartitekikuchiintro}, our key technical contribution is designing the matrix whose spectral norm upper bounds $\E_b[\val(\Psi_b^{(s)})]$ for each of the ``decomposed'' instances $\Psi^{(s)}$ for $2 \leq s \leq \frac{q+1}{2}$.
\end{enumerate}

We now formally describe the decomposition process.

\begin{lemma}[Hypergraph Decomposition]
\label{lem:decomp}
Let $H_1, \dots, H_k$ be $q$-uniform hypergraphs on $n$ vertices, and let $H$ be the multiset $H \coloneqq \cup_{i = 1}^k H_i$. 

For each $2 \leq s \leq \frac{q+1}{2}$, let $d_s$ be a positive integer such that $d_2 \geq d_s \geq \dots \geq d_{\frac{q+1}{2}} \geq 1$, and let $P_s \coloneqq \{Q \in {[n] \choose s} : \deg_{H}(Q) > d_s\}$. Then, there are $q$-uniform hypergraphs $H'_1, \dots, H'_k$ and, for each $2 \leq s \leq \frac{q+1}{2}$, bipartite hypergraphs $H_1^{(s)}, \dots, H_k^{(s)}$, with the following properties.
\begin{enumerate}[(1)]
\item Each $H^{(s)}_i$ is a bipartite hypergraph where each hyperedge contains $q - s$ left vertices in $[n]$ and one right vertex $p \in P_s$. Furthermore, $\abs{P_s} \leq O(\abs{H} / d_s)$.
\item Each $H'_i$ is a subset of $H_i$.
\item For each $i \in [k]$, there is a one-to-one correspondence between hyperedges $C \in H_i$ and the hyperedges in $H'_i, H_i^{(2)}, \dots, H_i^{(\frac{q+1}{2})}$ given by $(C, p) \in H_i^{(s)} \mapsto C \cup p \in H_i$ and $C \in H'_i \mapsto C \in H_i$.
\item Let $H' \coloneqq \cup_{i = 1}^k H'_i$. Then, for any $Q \in {[n] \choose s}$ with $2 \leq s \leq \frac{q+1}{2}$, it holds that $\deg_{H'}(Q) \leq d_s$.
\item If $H_i$ is a matching, then $H'_i$ and $H^{(s)}_i$ for $2 \leq s \leq \frac{q+1}{2}$ are also matchings.
\end{enumerate}
\end{lemma}
The proof of \cref{lem:decomp} follows by using a simple greedy algorithm, and is given in \cref{sec:decomp}.

Given the decomposition, the two main technical parts of the proof are given by the following two theorems. In the first theorem, we refute the $q$-XOR instance resulting from the hypergraph $H'$, and in the second theorem we refute the bipartite $(q-s+1)$-XOR instances from the hypergraphs $H^{(s)}$ for each $2 \leq s \leq \frac{q+1}{2}$.
\begin{restatable}[Refuting the regular $q$-XOR instance]{theorem}{regularref}
\label{thm:regularref}
Let $q \geq 3$ be an odd integer. Let $k, n$ be positive integers and $\delta \in (0,1)$. Let $\ell = \floor{n^{1 - 2/q} \cdot \delta^{-2/q}}$, and suppose that $k \geq 4\ell$. For $2 \leq t \leq \frac{q+1}{2}$, let $d_t \defeq \left(\frac{\ell}{n}\right)^{t - \frac{3}{2}} k$.

Let $H_1, \dots, H_k$ be $q$-uniform hypergraph matchings on $[n]$ of size $\leq \delta n$, and suppose that for every $Q \subseteq [n]$ with $2 \leq \abs{Q} \leq \frac{q+1}{2}$, it holds that $\deg_H(Q) \leq d_{\abs{Q}}$, where $H \defeq \cup_{i = 1}^k H_i$. Let $\Psi_b(x)$ be the polynomial in the variable $x_1, \dots, x_n$ defined as
\begin{equation*}
\Psi_b(x) = \sum_{i = 1}^k \sum_{C \in H_i} b_i \prod_{v \in C} x_v \mper
\end{equation*}
Then, $\E_{b \gets \Fits^k}[\val(\Psi_b)] \leq  O(n\sqrt{\delta k}) \cdot (k \ell \log n)^{1/4} $.
\end{restatable}

\begin{restatable}[Refuting the bipartite instances]{theorem}{bipartiteref}
\label{thm:bipartiteref}
Let $q \geq 3$ be an odd integer, and let $2 \leq s \leq \frac{q+1}{2}$. Let $k, n$ be positive integers and $\delta \in (0,1)$. Let $\ell = \floor{n^{1 - 2/q} \cdot \delta^{-2/q}}$, and suppose that $k \geq 4\ell$. For $2 \leq t \leq \frac{q+1}{2}$, let $d_t \defeq \left(\frac{\ell}{n}\right)^{t - \frac{3}{2}} k$. Let $P_s \subseteq {[n] \choose s}$ be a set with $4 \ell \leq \abs{P_s} \leq O\left(\frac{n k}{d_s}\right)$.

Let $H^{(s)}_1, \dots, H^{(s)}_k$ be bipartite $(q-s+1)$-uniform hypergraph matchings on ${[n] \choose q - s} \times P_s$ of size at most $\delta n$. Let $\Psi^{(s)}_b(x,y)$ be the polynomial in the variable $x_1, \dots, x_n$ and $\{y_p\}_{p \in P_s}$ defined as
\begin{equation*}
\Psi^{(s)}_b(x,y) = \sum_{i = 1}^k \sum_{(C, p) \in H_i} b_i y_p \prod_{v \in C} x_v \mper
\end{equation*}
Then, $\E_{b \gets \Fits^k}[\val(\Psi^{(s)}_b)] \leq  \delta n O(\sqrt{k \ell \log n})$.
\end{restatable}
We prove \cref{thm:regularref} in \cref{sec:regularref}, and we prove \cref{thm:bipartiteref} in \cref{sec:bipartiteref}.

With the above ingredients, we can now finish the proof of \cref{mthm:oddqldc}. 
\begin{proof}[Proof of \cref{mthm:oddqldc}]
By \cref{fact:LDCnormalform}, we may assume that our code $\Code$ is in LDC normal form, and our goal is to show that $k \leq O(\eps^{-4} \delta^{-2} \ell \log n)$ holds, where $\ell = \floor{n^{1 - 2/q} \cdot \delta^{-2/q}}$, which implies \cref{mthm:oddqldc}. 
 We will assume that $\delta$ satisfies $\delta \geq n^{-\frac{2}{q+2}}$, as otherwise $\delta^{-2} \ell \geq n$ holds, and so $k \leq n \leq O(\eps^{-4} \delta^{-2} \ell \log n)$ trivially holds.
 
For each $2 \leq s \leq \frac{q+1}{2}$, define $d_s \coloneqq \left(\frac{\ell}{n}\right)^{s - \frac{3}{2}} k$.  We will assume that $k \geq 4\ell$, as otherwise we are already done. Because of this assumption, we have $\left(\frac{\ell}{n}\right)^{\frac{q}{2} - 1} k \geq 1$, which implies that $d_s \geq 1$ for all $2 \leq s \leq \frac{q+1}{2}$.
We can thus apply \cref{lem:decomp} with these thresholds, which decomposes each $H_i$ into $H'_i$ and $H^{(s)}_i$ for $2 \leq s \leq \frac{q+1}{2}$. Note that $\abs{P_s} \leq O(\delta nk/d_s) = O(1) \cdot \left(\frac{n}{\ell}\right)^{s - \frac{3}{2}} \cdot \delta n$.

The one-to-one correspondence property in \cref{lem:decomp} implies that for each $b \in \Fits^k$ and every $x \in \Fits^n$, if we set $y_p = \prod_{v \in p} x_v$ for each $p \in P_s$ and $2 \leq s \leq \frac{q+1}{2}$, then it holds that $\Phi_b(x) = \Psi_b(x) + \sum_{s = 2}^{\frac{q+1}{2}} \Psi^{(s)}_{b}(x,y)$.

We can now apply \cref{thm:regularref,thm:bipartiteref} to bound $\E[\val(\Psi_b)]$ and $\E[\val(\Psi^{(s)}_b)]$. However, it is possible that the condition that $\abs{P_s} \geq 4 \ell$ does not hold. But, if $\abs{P_s} \leq 4 \ell$, then the conclusion of \cref{thm:bipartiteref} still holds. This is because for any $b$, $\val(\Psi^{(s)}_b) \leq \sum_{i = 1}^k \abs{H^{(s)}_i}$ trivially holds, and we also have $\sum_{i = 1}^k \abs{H^{(s)}_i} \leq \abs{P_s} d_s$, as each $p \in P_s$ contributes at most $d_s$ hyperedges to $\cup_{i = 1}^k H^{(s)}_i$. Hence, $\val(\Psi^{(s)}_b) \leq \ell d_s$ in this case, which is at most $\delta n O(\sqrt{k \ell \log n})$ when $\delta = \Omega(n^{-\frac{2}{q+2}})$.

We thus have that
\begin{flalign*}
\eps \delta n k \leq \E[\val(\Phi_b)] \leq \E[\val(\Psi_b)] + \sum_{s = 2}^{\frac{q + 1}{2}} \E[\val(\Psi^{(s)}_b)] \leq O(n\sqrt{\delta k}) \cdot (k \ell \log n)^{1/4}  + \delta n O(\sqrt{k \ell \log n}) \mper
\end{flalign*}
We have two cases. If $O(n\sqrt{\delta k}) \cdot (k \ell \log n)^{1/4}$ is larger than $\delta n O(\sqrt{k \ell \log n})$, then we conclude that
\begin{flalign*}
\eps \delta n k \leq O(n\sqrt{\delta k}) \cdot (k \ell \log n)^{1/4}  \implies  k \leq O(\eps^{-4} \delta^{-2} \ell \log n) \mcom
\end{flalign*}
and if $O(n\sqrt{\delta k}) \cdot (k \ell \log n)^{1/4}$ is smaller than $\delta n O(\sqrt{k \ell \log n})$, we conclude that
\begin{flalign*}
\eps \delta n k \leq \delta n O(\sqrt{k \ell \log n}) \implies k \leq O(\eps^{-1} \ell \log n) \mper
\end{flalign*}
Thus, we have $k \leq O(\eps^{-4} \delta^{-2} \ell \log n) = O(n^{1 - \frac{2}{q}} \delta^{-2- \frac{2}{q}}\eps^{-4} \log n)$, which finishes the proof.
\end{proof}

\subsection{Hypergraph decomposition: proof of \cref{lem:decomp}}
\label{sec:decomp}

We prove \cref{lem:decomp} by analyzing the following greedy algorithm.
\begin{mdframed}
  \begin{algorithm}
    \label{alg:decomp}\mbox{}
    \begin{description}
    \item[Given:]
       $q$-uniform hypergraphs $H_1, \dots, H_k$ and parameters $d_2 \geq d_3 \geq \dots \geq d_{\frac{q+1}{2}} \geq 1$.
    \item[Output:]
       $q$-uniform hypergraphs $H'_1, \dots, H'_k$ and for $2 \leq s \leq \frac{q+1}{2}$, bipartite $(q - s + 1)$-uniform hypergraphs $H^{(s)}_1, \dots, H^{(s)}_k$ over the left vertex set $[n]$ and right vertex set $P_s \subseteq {[n] \choose s}$.
           \item[Operation:]\mbox{}
    \begin{enumerate}
    	\item \textbf{Initialize:} $H'_i = H_i$ for all $i \in [k]$, $P_s = \emptyset$, and $P'_s = \{Q \in {[n] \choose s} : \deg_{H'}(Q) > d_s\}$, where $H' = \cup_{i \in [k]} H'_i$.
	\item \textbf{For $t = \frac{q+1}{2}, \dots, 1$:}
	\begin{enumerate}[(1)]
	    \item \textbf{While $P'_t$ is nonempty:}
			\begin{enumerate}[(a)]
    			\item Choose $p \in P'_t$ arbitrarily.
			\item Choose an arbitrary set of $d_t+1$ hyperedges in $H'$ containing the set $p$. Namely, let $C_{i_1}, \dots, C_{i_{d_t+1}}$ be hyperedges in $H'$ where $C_{i_1} \in H'_{i_1}$, \dots, $C_{i_{d_t+1}} \in H'_{i_{d_t+1}}$.	
    			\item Add $p$ to $P_t$ and for each $r \in [d_{t} + 1]$, remove $C_{i_r}$ from $H'_{i_r}$ and add the hyperedge $(C_{i_r} \setminus p, p)$ to $H^{(t)}_{i_r}$. 
			\item Recompute $P'_t = \{Q \in {[n] \choose t} : \deg_{H'}(Q) > d_t\}$.
    		\end{enumerate}
		\end{enumerate}
		\item Output $H'_1, \dots, H'_k$ and $H^{(s)}_1, \dots, H^{(s)}_k$ for all $2 \leq s \leq \frac{q+1}{2}$.
		\end{enumerate}
    \end{description}
  \end{algorithm}
  \end{mdframed}
  We now need to show that the output of \cref{alg:decomp} has the desired properties.
  
  Item (1) holds by construction, as each $p \in P_s$ has size $s$ so when the hyperedge $C$ is split into $(C \setminus p, p)$, $\abs{C \setminus p} = q -s$. We have that $\abs{P_s} \leq O(\abs{H} / d_s)$, as each hyperedge $C \in H$ has (crudely) at most $2^q = O(1)$ subsets of size exactly $s$, and each $p \in P_s$ must appear at least $d_s + 1$ times across hyperedges in $H$.
  
  Item (2) holds by construction, as we start with $H'_i = H_i$ and only remove edges from $H'_i$.
  
  Item (3) holds because each hyperedge $C \in H_i$ is either never removed (in which case it appears in $H'_i$), or it is removed exactly once. If it is removed by choosing some $p \in P_s$, then it appears in $H^{(s)}_i$ as the hyperedge $(C \setminus p, p)$.
  
  Item (4) holds because otherwise the algorithm would not have terminated.
  
  Item (5) holds because the operations done by \cref{alg:decomp} do not affect the matching property. This finishes the proof.
  
\section{Refuting the Regular $q$-XOR Instance}
\label{sec:regularref}
In this section, we prove \cref{thm:regularref}, which we recall below.

\regularref*

The proof of \cref{thm:regularref} follows the overall blueprint outlined in the work of~\cite{AlrabiahGKM23}, as explained in \cref{sec:agkmodd}.

\parhead{Step 1: the Cauchy-Schwarz trick.} First, we show that we can relate $\Psi(x)$ to a certain ``Cauchy-Schwarzed'' polynomial $f_{L,R}(x)$.
\begin{lemma}[Cauchy-Schwarz Trick]
\label{lem:cauchyschwarz}
Let $\Psi$ be as in \cref{thm:regularref} and let $L,R \subseteq [k]$ be a random partition of $[k]$, i.e., each $i \in [k]$ appears in $L$ with probability $1/2$, independently, and $R = [k] \setminus L$. Let $f_{L,R}(x)$ be the polynomial defined as
\begin{equation*}
f_{L,R}(x) \defeq \sum_{i \in L, j \in R} \sum_{u \in [n]} \sum_{(u,C_1) \in H_i, (u,C_2) \in H_j} b_i b_j \prod_{v \in C_1} x_v \prod_{v \in C_2} x_v \mper
\end{equation*}
Then, it holds that $(q \val(\Psi))^2 \leq  q \delta n^2 + 4n \E_{(L,R)} \val(f_{L,R})$. In particular, $\E_{b \in \Fits^k} [q^2 \cdot \val(\Psi)^2 ] \leq q \delta n^2 + 4 n \E_{(L,R)} \E_{b \in \Fits^k} [\val(f_{L,R})]$.
\end{lemma}
\begin{proof}
Fix any assignment to $x \in \Fits^n$. We have that
\begin{flalign*}
&(q\Psi(x))^2 = \left(\sum_{u \in [n]} x_u \sum_{i \in [k]} \sum_{(u, C) \in H_i}  b_i x_C\right)^2 \leq \left(\sum_{u \in [n]} x_u^2\right) \left( \sum_{u \in [n]} \left(\sum_{i \in [k]} \sum_{(u, C) \in H_i} b_i x_C\right)^2 \right) \\
&= n \sum_{u \in [n]} \sum_{i, j \in [k]} \sum_{\substack{(u, C_1) \in H_i \\ (u, C_2) \in H_j}} b_i b_j x_{C_1} x_{C_2} = n\left(q\sum_{i \in [k]} \abs{H_i} + \sum_{u \in [n]}  \sum_{i, j \in [k], i \ne j}\sum_{\substack{(u, C_1) \in H_i \\ (u, C_2) \in H_j}} b_i b_j x_{C_1} x_{C_2} \right) \\
&= q n \cdot \delta n + 4 n \cdot \E_{(L, R)} f_{L,R}(x) \enspace,
\end{flalign*}
where the first equality is because there are $q$ ways to decompose a set $C_i \in H_i$ with $\abs{C_i} = q$ into a pair $(u, C)$ with $\abs{C} = q - 1$, the inequality follows by the Cauchy-Schwarz inequality, and the last equality follows because for a pair of hypergraphs $H_i$ and $H_j$, we have $i \in L$ and $j \in R$ with probability $1/4$. Finally, $\max_{x \in \{-1,1\}^n} \E_{(L,R)} f_{L,R}(x) \leq \E_{(L,R)} [\max_{x \in \{-1,1\}^n} f_{L,R}(x)] = \E_{(L,R)} \val(f_{L,R})$. Thus, we have that $q^2 \cdot \val(\Psi)^2 \leq q \delta n^2 + 4 n \cdot  \E_{(L,R)} \val(f_{L,R})$. 
\end{proof}

\parhead{Step 2: defining the Kikuchi matrices.}
Next, we define the Kikuchi matrices that we will use and relate them to the polynomial $f_{L,R}$.

\begin{definition}
\label{def:regularkikuchi}
Let $q \geq 3$ be an odd integer and let $\ell = \floor{n^{1 - 2/q} \cdot \delta^{-2/q}}$. Let $(u,C_1)$ be a hyperedge with $\abs{C_1} = q - 1$ and let $(u,C_2)$ be a hyperedge with $\abs{C_2} = q - 1$. We define the matrix $A_{u, C_1, C_2}$ to be the matrix indexed by pairs of sets $(S_1, S_2)$ where $S_1, S_2 \subseteq [n]$ and $\abs{S_1} = \abs{S_2} = \ell$, where $A_{u, C_1, C_2}((S_1, S_2), (T_1, T_2)) = 1$ if $S_1 \oplus T_1 = C_1$ and $S_2 \oplus T_2 = C_2$, and $0$ otherwise. We note that this is equivalent to $\abs{S_1 \cap C_1} = \abs{T_1 \cap C_1} = \frac{q - 1}{2}$ and $\abs{S_2 \cap C_2} = \abs{T_2 \cap C_2} = \frac{q - 1}{2}$.

We will also view the matrix $A_{u,C_1, C_2}$ as the adjacency matrix of a graph $G_{u, C_1, C_2}$.

For $i \ne j \in [k]$, we define $A_{i,j} \defeq \sum_{u \in [n]} \sum_{(u,C_1) \in H_i, (u,C_2) \in H_j} A_{u, C_1, C_2}$. We also define $A_{i} \defeq \sum_{j \in R} b_j A_{i,j}$ and $A \defeq \sum_{i \in L} b_i A_i$.
\end{definition}

\begin{claim}
\label{claim:regularsoundness}
For each $(u,C_1, C_2)$, the matrix $A_{u,C_1,C_2}$ defined in \cref{def:regularkikuchi} satisfies the following properties.
\begin{enumerate}[(1)]
\item The matrix $A_{u, C_1, C_2}$ has exactly $D = {q - 1 \choose \frac{q-1}{2}}^2 {n - (q - 1) \choose \ell - \frac{q - 1}{2}}^2$ nonzero entries.
\item For each $x \in \Fits^n$, let $z \in \Fits^{{n \choose \ell}^2}$ be defined as: $z_{S_1, S_2} := \prod_{v \in S_1} x_v \prod_{v \in S_2} x_v$. Then, $z^{\top} A_{u, C_1, C_2} z = D \prod_{v \in C_1} x_v \prod_{v \in C_2} x_v$.
\end{enumerate}
In particular, $z^{\top} A z = D f_{L,R}(x)$. 

As additional notation, we let $N \defeq {n \choose \ell}^2$ and $d = \frac{\delta n k D}{N}$.
\end{claim}
\begin{proof}
To prove Item (1), we will count the number of edges. By definition, we have an edge $((S_1, S_2), (T_1, T_2))$ in the graph with adjacency matrix $A_{u,C_1,C_2}$ iff $\abs{S_1 \cap C_1} = \frac{q-1}{2}$ and $\abs{S_2 \cap C'_2} = \frac{q - 1}{2}$. The number of such $S_1$ is ${q - 1 \choose \frac{q-1}{2}} {n - (q - 1) \choose \ell - \frac{q - 1}{2}}$, and the number of such $S_2$ is the same. Hence, Item (1) holds.

To prove Item (2), we observe that
\begin{flalign*}
&z^{\top} A_{u, C_1, C_2} z = \sum_{((S_1, S_2), (T_1, T_2)) \in E(G_{u, C_1, C_2})} z_{S_1, S_2} w_{T_1, T_2} = \sum_{((S_1, S_2), (T_1, T_2)) \in E(G_{u, C_1, C_2})} \prod_{v \in S_1} x_v \prod_{v \in S_2} x_v \prod_{v \in T_1} x_v \prod_{v \in T_2} x_v \\
&=\sum_{((S_1, S_2), (T_1, T_2)) \in E(G_{u, C_1, C_2})} \prod_{v \in S_1 \oplus T_1} x_v \prod_{v \in S_2 \oplus T_2} x_v = \sum_{((S_1, S_2), (T_1, T_2)) \in E(G_{u, C_1, C_2})} \prod_{v \in {C_1}} x_v  \prod_{v \in {C_2}} x_v = D \prod_{v \in {C_1}} x_v  \prod_{v \in {C_2}} x_v \mper
\end{flalign*}

The ``in particular'' follows immediately from Item (2) and the definition of $A$.
\end{proof}

\parhead{Step 3: finding an approximately regular submatrix.}
The key technical lemma, which we shall prove in \cref{sec:regularrowpruning}, shows that we can find an approximately biregular subgraph of $A_i$ for each $i \in L$.
\begin{lemma}[Approximately regular submatrix]
\label{lem:regularrowpruning}
For $i \in L$, let $A_i$ be defined as in \cref{def:regularkikuchi}. There exists a positive integer $D'$ with $D \geq D' \geq \frac{D}{2}$ such that the following holds. For each $i \in L$, $(u, C_1) \in H_i$, $j \in R$, and $(u,C_2) \in H_j$, there exists a matrix $B_{i, u, C_1, C_2} \in \Bits^{{n \choose \ell}^2 \times {n \choose \ell}^2}$ with the following properties:
\begin{enumerate}[(1)]
\item $B_{i, u, C_1, C_2}$ is a ``subgraph'' of $A_{u, C_1, C_2}$. Namely, $B_{i, u, C_1, C_2} = B_{i, u, C_1, C_2}^{\top}$ and if $B_{i, u, C_1, C_2}((S_1, S_2), (T_1, T_2)) = 1$, then we also have $A_{u, C_1, C_2}((S_1, S_2), (T_1, T_2)) = 1$.
\item $B_{i,u, C_1, C_2}$ has exactly $D'$ nonzero entries.
\item The matrix $B_i \defeq \sum_{j \in R} b_j \sum_{u \in [n]} \sum_{(u, C_1) \in H_i, (u, C_2) \in H_j} B_{i, u, C_1, C_2}$ has at most $O(d)$ nonzero entries per row or column, where $d = \frac{\delta nk D}{N}$.
\end{enumerate}
\end{lemma}

\parhead{Step 4: finishing the proof.}
With \cref{lem:regularrowpruning} in hand, we can now finish the proof. Let $B_i$ be the matrix defined in \cref{lem:regularrowpruning}. We observe that Items (1), (2), and (3) in \cref{lem:regularrowpruning}, along with \cref{claim:regularsoundness}, imply that for each $x \in \Fits^n$, there exists $z \in \Fits^{N}$ such that $z^{\top} B z = D' f_{L,R}(x)$. Hence, for any $x \in \Fits^n$, it holds that $\val(f_{L,R}) \leq \norm{B}_2 \cdot N$. We thus have that $\E_b[\val(f_{L,R})] \leq \frac{N}{D'} \E_b[\norm{B}_2]$.

It remains to bound $\E_b[\norm{B}_2]$, which we do using Matrix Khintchine (\cref{fact:matrixkhintchine}) in the following claim.
\begin{claim}
\label{claim:specnormregular}
Let $B$ be the matrix defined in \cref{lem:regularrowpruning}. Then, $\E_b[\norm{B}_2] \leq O(d \sqrt{k \ell \log n})$.
\end{claim}

We postpone the proof of \cref{claim:specnormregular} to the end of this section, and finish the proof of \cref{thm:regularref}. We have that
\begin{flalign*}
&\E_b[\val(f_{L,R})] \leq \frac{N}{D'} \E_b[\norm{B}_2]  \leq  \frac{2N}{D}O(d\sqrt{k \ell \log n}) \\
&= O(1) \cdot \frac{N d}{D} \sqrt{ k \ell \log n} = \delta n k \cdot O(\sqrt{k \ell \log n}) \mcom
\end{flalign*}
where we recall that $d = \delta n k D/N$. We note that the above holds for \emph{any} choice of the partition $L \cup R = [k]$. Finally, we recall that by \cref{lem:cauchyschwarz}, we have that
\begin{flalign*}
&(\E_b[q\val(\Psi)])^2 \leq \E_{b \in \Fits^k} [q^2 \cdot \val(\Psi)^2 ] \leq q \delta n^2 + 4 n \E_{(L,R)} \E_{b \in \Fits^k} [\val(f_{L,R})] \leq q \delta n^2 + \delta n^2 k \cdot O(\sqrt{k \ell \log n}) \\
&\implies \E_b[\val(\Psi)] \leq  O(n\sqrt{\delta k}) \cdot (k \ell \log n)^{1/4} \mper
\end{flalign*}

We now finish the proof of \cref{claim:specnormregular}.
\begin{proof}[Proof of \cref{claim:specnormregular}]
 By Matrix Khintchine (\cref{fact:matrixkhintchine}), we have $\E_b[\norm{B}_2] \leq O(\sqrt{\sigma^2 \log N})$, where $\sigma^2 = \norm{\sum_{i = 1}^k B_i^2}$, as $B_i$ is symmetric. Since $B_i$ is symmetric, $\norm{B_i}_2$ is bounded by the maximum $\ell_1$-norm of a row in this matrix. By construction of $B_i$, this is $O(d)$. Hence, $\sigma^2 \leq k \cdot O(d)^2 = O(kd^2)$.
 
 We can thus set $\sigma^2 = O(k d^2)$ and apply \cref{fact:matrixkhintchine} to conclude that $\E_b[\norm{B}_2] \leq O(d\sqrt{k \log N})$. Recall that we have $N = {n \choose \ell} \leq n^{\ell} (nk)^{\ell} \leq n^{O(\ell)}$. Hence, $\log N = O(\ell \log n)$, which finishes the proof.
 \end{proof}

\subsection{Finding an approximately regular subgraph: proof of \cref{lem:regularrowpruning}}
\label{sec:regularrowpruning}
In this section, we prove \cref{lem:regularrowpruning}. We will prove \cref{lem:regularrowpruning} by using the strategy, due to ~\cite{Yankovitz24}, of bounding ``conditional first moments''.  These moment bounds form the main technical component of the argument.
\begin{lemma}[Conditional first moment bounds]
\label{lem:regularmoments}
Fix $i \in L$. For a vertex $(S_1, S_2)$, let $\deg_{i}(S_1, S_2)$ denote the degree of $(S_1, S_2)$ in $A_i$.

Let $(u,C_1) \in H_i$ and $(u,C_2) \in \cup_{j \in R} H_j$. Let $\mu_{u, C_1, C_2}$ denote the distribution over vertices that first chooses a uniformly random edge in $A_{u,C_1, C_2}$ and then outputs a random endpoint. Then, it holds that
\begin{flalign*}
&\E_{(S_1, S_2) \sim \mu_{u, C_1, C_2}}[\deg_{i}(S_1, S_2)] \leq 1 + O(1) \left(\frac{\ell}{n} \right)^{q-1}\delta nk \mper
\end{flalign*}
\end{lemma}

\begin{claim}[Degree bound]
\label{claim:regulardegbound}
Let $d = \frac{\delta nk  D}{N}$. Then, we have that $d \geq \Omega(1) \cdot \left(\frac{\ell}{n} \right)^{q-1}\delta nk$ and that $\left(\frac{\ell}{n} \right)^{q-1}\delta nk \geq \Omega(1)$.
\end{claim}
\begin{proof}[Proof of \cref{claim:regulardegbound}]
Applying \cref{fact:binomest}, we have that
\begin{flalign*}
\frac{\delta n D}{N} = \delta n \cdot \left( \frac{{q - 1 \choose \frac{q-1}{2}}{n - (q-1) \choose \ell - \frac{q -1}{2}}}{{n \choose \ell}} \right)^2 \geq \Omega(1) \cdot \delta n \cdot \left(\frac{\ell}{n}\right)^{q - 1} \mper
\end{flalign*}
Because $\ell = \floor{n^{1 - 2/q} \delta^{-2/q}}$ and $k \geq \ell$, we have $\left(\frac{\ell}{n}\right)^{q-1} \delta n k \geq 1$, which finishes the proof.
\end{proof}

We postpone the proof of \cref{lem:regularmoments} to the end of this subsection, and now use it to finish the proof of \cref{lem:regularrowpruning}.

\begin{proof}[Proof of \cref{lem:regularrowpruning} from \cref{lem:regularmoments}]
Fix $i \in [k]$. Let $\Gamma$ be a constant (to be chosen later), and let $V'_i = \{(S_1, S_2) : \deg_{i}(S_1, S_2) \leq \Gamma d\}$.

Let $(u,C_1) \in H_i$ and $(u,C_2) \in \cup_{j \in R} H_j$. We let $A'_{i, u,C_1, C_2}$ be the matrix where $A'_{i, u, C_1, C_2}((S_1, S_2), (T_1, T_2)) = A_{u, C_1, C_2}((S_1, S_2), (T_1, T_2))$ if $(S_1, S_2) \in V'$ and $(T_1, T_2) \in V'$, and otherwise $A'_{i, u, C_1, C_2}((S_1, S_2), (T_1, T_2)) =  0$. Namely, we have ``zeroed out'' all rows and columns of $A_{u, C_1, C_2}$ that are not in $V'$. Notice that $A'_{i, u, C_1, C_2}$ depends on $i \in L$ because $V'$ does.

The conditional moment bound from \cref{lem:regularmoments}, combined with the lower bound on $d$ from \cref{claim:regulardegbound} implies that $\E_{(S_1, S_2) \sim \mu_{u, C_1, C_2}}[\deg_{i}(S_1, S_2)] \leq O(d)$. Hence, applying Markov's inequality, the number of vertices $(S_1, S_2)$ that are adjacent to an edge labeled by $(u, C_1, C_2)$ and have $\deg_{i}(S_1, S_2) > \Gamma d$ is at most $O(D/\Gamma)$. Hence, there must be at least $D (1 - O(1/\Gamma))$ edges, i.e., nonzero entries, in $A'_{i, u,C_1,C_2}$.

Now, we let $B_{i, u, C_1, C_2}$ be any subgraph of $A'_{i, u, C_1, C_2}$ where $B_{i, u, C_1, C_2}$ has \emph{exactly} $D' = \floor{D(1 - O(1/\Gamma))}$ edges. This can be achieved by simply removing edges if there are too many. By choosing $\Gamma$ to be a sufficiently large constant, we ensure that $D' \geq D/2$. Note that because $B_{i, u, C_1, C_2}$ is the adjacency matrix of a graph, it is a symmetric matrix.

To prove the third property, we observe that for any vertex $(S_1, S_2)$, the matrix 
\begin{equation*}
B_i = \sum_{u \in [n]} \sum_{(u, C_1) \in H_i, (u, C_2) \in \cup_{j \in R} H_j} B_{i,u, C_1, C_2}
\end{equation*}
 has at most $\Gamma d = O(d)$ nonzero entries in the $(S_1, S_2)$-th row or column. Indeed, this follows because it is a subgraph of the original graph $A_i$, and if $(S_1, S_2)$ had degree $> \Gamma d_L$ in $A_i$ then it has degree $0$ in $B_i$. This finishes the proof.\end{proof}

It remains to prove \cref{lem:regularmoments}, which we do now.
\begin{proof}[Proof of \cref{lem:regularmoments}]
Let $(u,C_1) \in H_i$ and $(u, C_2) \in H_j$ for some $j \in R$. We observe that for any $(S_1, S_2) \in {[n] \choose \ell} \times {[n] \choose \ell}$, the vertex $(S_1, S_2)$ is adjacent to at most one edge labeled by $(u, C_1, C_2)$. Hence, it follows that $\mu_{u, C_1, C_2}$ is uniform over pairs of sets $(S_1, S_2)$ such that $\abs{S_1 \cap C_1} = \frac{q - 1}{2}$ and $\abs{S_2 \cap C_2} = \frac{q - 1}{2}$. Thus,
\begin{flalign*}
&\E_{(S_1, S_2) \sim \mu_{u, C_1, C_2}}[\deg_{i}(S_1, S_2)] \leq 1 + \frac{1}{D} \sum_{(u', C'_1, C'_2)} \abs{\{(S_1, S_2) :  \abs{S_1 \cap C_1} = \abs{S_1 \cap C'_1} = \abs{S_2 \cap C_2} = \abs{S_2 \cap C'_2} = \frac{q - 1}{2}\}} \\
&\leq 1 + \frac{1}{D}\sum_{\substack{Z_1 \subseteq C_1 \\ Z_2 \subseteq C_2 \\ \abs{Z_1} = \abs{Z_2} = \frac{q-1}{2}}}\sum_{(u', C'_1, C'_2)} \abs{\{(S_1, S_2) : \substack{Z_1 \subseteq S_1 \\ Z_2 \subseteq S_2}, \substack{\abs{S_1 \cap C'_1} =  \frac{q-1}{2} \\ \abs{S_2 \cap C'_2} = \frac{q - 1}{2}}\}} \quad \text{(because $Z_1 \subseteq S_1$ implies $\abs{S_1 \cap C_1} \geq \frac{q-1}{2}$)}\\
&\leq 1 + \frac{1}{D}\sum_{\substack{Z_1 \subseteq C_1 \\ Z_2 \subseteq C_2 \\ \abs{Z_1} = \abs{Z_2} = \frac{q-1}{2}}} \sum_{\substack{Q_1 \subseteq Z_1 \\ Q_2 \subseteq Z_2}} \sum_{\substack{(u', C'_1, C'_2) \\ Q_1 = C'_1 \cap Z_1 \\ Q_2 = C'_2 \cap Z_2}} \abs{\{(R_1, R_2) : \substack{\abs{R_1} = \ell - \abs{Z_1} \\ \abs{R_2} = \ell - \abs{Z_2}}, \substack{\abs{R_1 \cap C'_1} =  \frac{q - 1}{2} - \abs{Q_1} \\ \abs{R_2 \cap C'_2} = \frac{q - 1}{2} - \abs{Q_2}}\}} \quad \text{(by taking $R_1 = S_1 \setminus Z_1$)}  \\
&\leq1 + \frac{1}{D}\sum_{\substack{Z_1 \subseteq C_1 \\ Z_2 \subseteq C_2 \\ \abs{Z_1} = \abs{Z_2} = \frac{q-1}{2}}} \sum_{\substack{Q_1 \subseteq Z_1 \\ Q_2 \subseteq Z_2}} \sum_{\substack{(u', C'_1, C'_2) \\ Q_1 = C'_1 \cap Z_1 \\ Q_2 = C'_2 \cap Z_2}}  {q - 1\choose \frac{q - 1}{2} - \abs{Q_1}} {n - (q-1) \choose \ell - \abs{Z_1} - (\frac{q - 1}{2} - \abs{Q_1})} {q - 1 \choose \frac{q - 1}{2} - \abs{Q_2}} {n - (q-1) \choose \ell - \abs{Z_2} - (\frac{q - 1}{2} - \abs{Q_2})} \\
&\leq 1 + \frac{O(1)}{D}\sum_{\substack{Z_1 \subseteq C_1 \\ Z_2 \subseteq C_2 \\ \abs{Z_1} = \abs{Z_2} = \frac{q-1}{2}}} \sum_{\substack{Q_1 \subseteq Z_1 \\ Q_2 \subseteq Z_2}} \sum_{\substack{(u', C'_1, C'_2) \\ Q_1 = C'_1 \cap Z_1 \\ Q_2 = C'_2 \cap Z_2}}   {n \choose \ell - \abs{Z_1} - (\frac{q - 1}{2} - \abs{Q_1})}{n \choose \ell - \abs{Z_2} - (\frac{q - 1}{2} - \abs{Q_2})} 
\end{flalign*}

Recall that $D \defeq {q - 1 \choose \frac{q-1}{2}}^2 {n - (q-1) \choose \ell - \frac{q -1}{2}}^2$. By \cref{fact:binomest}, we have
\begin{flalign*}
&\frac{1}{D} {n \choose \ell - \abs{Z_1} - (\frac{q - 1}{2} - \abs{Q_1})}{n \choose \ell - \abs{Z_2} - (\frac{q - 1}{2} - \abs{Q_2})} \leq O(1) \cdot \left(\frac{\ell}{n} \right)^{\abs{Z_1} - \abs{Q_1} + \abs{Z_2} - \abs{Q_2}} = O(1) \cdot \left(\frac{\ell}{n} \right)^{(q-1) - \abs{Q_1} - \abs{Q_2}} \mcom
\end{flalign*}
as $\abs{Z_1} = \abs{Z_2} = \frac{q-1}{2}$.

For sets $Q_1, Q_2$, we let $\mu(Q_1, Q_2)$ be the number of $(u', C'_1, C'_2)$ such that $Q_1 \subseteq C'_1$ and $Q_2 \subseteq C'_2$. We then have that
\begin{flalign*}
&\E_{(S_1, S_2) \sim \mu_{u, C_1, C_2}}[\deg_{i}(S_1, S_2)] \leq 1 + O(1) \sum_{\substack{Z_1 \subseteq C_1 \\ Z_2 \subseteq C_2 \\ \abs{Z_1} = \abs{Z_2} = \frac{q-1}{2}}} \sum_{\substack{Q_1 \subseteq Z_1 \\ Q_2 \subseteq Z_2}} \mu(Q_1, Q_2)\left(\frac{\ell}{n} \right)^{(q-1) - \abs{Q_1} - \abs{Q_2}} \mper
\end{flalign*}

We now show that $\mu(Q_1, Q_2) \leq O(1) \left(\frac{\ell}{n} \right)^{\abs{Q_1} + \abs{Q_2}} \delta nk$ where $0 \leq \abs{Q_1}, \abs{Q_2} \leq \frac{q - 1}{2}$. We have several cases.
\begin{enumerate}[(1)]
\item $\abs{Q_1} = 0$. In this case, we know that $(u',C'_2)$ is in $H_j$ for some $j \in R$ with $Q_2 \subseteq C'_2$. We have three subcases.
\begin{enumerate}[(a)]
\item $\abs{Q_2} = 0$. Then, there are at most $q \delta n$ choices for $(u',C'_1) \in H_i$, as $\abs{H_i} = \delta n$ and we have $q$ choices for the special element $u$. Furthermore, given $u$, there are at most $k$ choices for $(u',C'_2)$ with $(u',C'_2) \in \cup_{j \in R} H_j$, as each $H_j$ is a matching and $\abs{R} \leq k$. Thus, $\mu(Q_1, Q_2) \leq O(\delta n k)$ in this case, which satisfies the desired bound as $\abs{Q_1} + \abs{Q_2} = 0$.
\item $\abs{Q_2} = 1$. Then, there are at most $q k$ choices for $(u', C'_2) \in \cup_{j \in R} H_j$ with $Q_2 \subseteq C'_2$. Indeed, this is because each $H_j$ is matching, and $\abs{R} \leq k$. As $H_i$ is a matching, there is at most one choice for $(u',C'_1) \in H_i$. Hence, $\mu(Q_1, Q_2) \leq O(k)$ in this case, which is $\leq O(1) \left(\frac{\ell}{n} \right) \delta nk$.
\item $\abs{Q_2} \geq 2$. Then, there are at most $q d_{\abs{Q_2}}$ choices for $(u',C'_2) \in \cup_{j \in R} H_j$, by regularity. As before, given $(u,C_2)$, there is at most one choice for $C'_1 \in H_i$. Hence, $\mu(Q_1, Q_2) \leq O(d_{\abs{Q_2}})$ in this case. Since $2 \leq \abs{Q_2} \leq \frac{q - 1}{2}$, we have that $d_{\abs{Q_2}} =  \left(\frac{\ell}{n} \right)^{\abs{Q_2} - \frac{3}{2}} k$, which is at most $\left(\frac{\ell}{n} \right)^{\abs{Q_2}} \delta nk$ since $\left(\frac{\ell}{n}\right)^{\frac{3}{2}} \delta n \geq 1$.
\end{enumerate}
\item $\abs{Q_1} \geq 1$. Then, there are at most $q$ choices for $(u',C'_1) \in H_i$. It then follows that we have determined $\abs{Q_2} + 1$ elements of $(u', C'_2) \in \cup_{j \in R} H_j$, namely $u'$ along with $Q_2$. We have two subcases.
\begin{enumerate}[(a)]
\item $\abs{Q_2} = 0$. Then, we have determined one element of $(u', C'_2)$, and so we have at most $k$ choices. Thus, $\mu(Q_1, Q_2) \leq O(k)$ in this case, which is at most $O(1) \left(\frac{\ell}{n} \right)^{\abs{Q_1}} \delta nk$, since $\abs{Q_1} \leq \frac{q-1}{2}$. 
\item $\abs{Q_2} \geq 1$. Then, we have determined at least two elements of $(u', C'_2)$. As $\abs{Q_2} \leq \frac{q-1}{2}$, we have that $\abs{Q_2} + 1 \leq \frac{q+1}{2}$, and so we have at most $d_{\abs{Q_2} + 1}$ choices in this case. Thus, $\mu(Q_1, Q_2) \leq O(d_{\abs{Q_2} + 1})$. As $d_{\abs{Q_2} + 1} = \left(\frac{\ell}{n} \right)^{\abs{Q_2} - \frac{1}{2}} k \leq \left(\frac{\ell}{n} \right)^{\frac{q - 1}{2} + \abs{Q_2}} \delta nk \leq \left(\frac{\ell}{n} \right)^{\abs{Q_1} + \abs{Q_2}}  \delta nk$, where we use that $\abs{Q_1} \leq \frac{q - 1}{2}$, we again have the desired bound on $\mu(Q_1, Q_2)$.
\end{enumerate}
\end{enumerate}
We have thus shown that $\mu(Q_1, Q_2) \leq O(1) \left(\frac{\ell}{n} \right)^{\abs{Q_1} + \abs{Q_2}} \delta nk$. Hence, 
\begin{flalign*}
&\E_{(S_1, S_2) \sim \mu_{u, C_1, C_2}}[\deg_{i}(S_1, S_2)] \leq 1 + O(1) \sum_{\substack{Z_1 \subseteq C_1 \\ Z_2 \subseteq C_2 \\ \abs{Z_1} = \abs{Z_2} = \frac{q-1}{2}}} \sum_{\substack{Q_1 \subseteq Z_1 \\ Q_2 \subseteq Z_2}} \left(\frac{\ell}{n} \right)^{\abs{Q_1} + \abs{Q_2}} \delta nk \cdot \left(\frac{\ell}{n} \right)^{(q-1) - \abs{Q_1} - \abs{Q_2}} \\
&\leq 1 + O(1) \left(\frac{\ell}{n} \right)^{q-1}\delta nk \mcom
\end{flalign*}
which finishes the proof.
\end{proof}

\section{Refuting the Bipartite Instances}
\label{sec:bipartiteref}
In this section, we prove \cref{thm:bipartiteref}, which we recall below.
\bipartiteref*

For notational simplicity, we will assume that $n^{1 - 2/q} \cdot \delta^{-2/q}$ is an integer, so that $\ell = n^{1 - 2/q} \cdot \delta^{-2/q}$. We note that if $n^{1 - 2/q} \cdot \delta^{-2/q}$ is not an integer, then we can set $\ell = \floor{n^{1 - 2/q} \cdot \delta^{-2/q}}$, and this only changes the bounds of the following proof by an $O(1)$-factor.
We will also write $H_i$ instead of $H_i^{(s)}$ to simplify notation.

We begin by defining the following Kikuchi matrix. 
\begin{definition}[Kikuchi matrix for bipartite hypergraphs]
\label{def:bipartitekikuchi}
For each $C \in {[n] \choose q-s}$ and $p \in P_s$, we define the bipartite graph $G_{C,p}$, parametrized by $\ell$, as follows. The left vertex set $V_L$ is the set of pairs of sets $(S_1, S_2)$ where $S_1 \in {[n] \choose \ell}$ and $S_2 \in {P_s \choose \ell}$. The right vertex set $V_R$ is the set of pairs of sets $(T_1, T_2)$ where $T_1 \in {[n] \choose \ell + 1 - s}$ and $T_2 \in {P_s \choose \ell + 1}$. We add an edge $((S_1, S_2), (T_1, T_2))$, which we view as ``labeled'' by $C$, if the following conditions hold:
\begin{enumerate}[(1)]
\item $S_1 \oplus T_1 = C$ and $S_2 \oplus T_2 = \{p\}$;
\item $\abs{S_1 \cap C} = \frac{q-1}{2}$ (and so $\abs{T_1 \cap C} = \frac{q + 1}{2} - s$).
\end{enumerate}
We can naturally view the graph $G_{C,p}$ as corresponding to its bipartite adjacency matrix $A_{C,p} \in \Bits^{V_L \times V_R}$. For each $i \in [k]$, we define the matrix $A_i = \sum_{(C,p) \in H_i} A_{C,p}$. We let $A = \sum_{i = 1}^k b_i A_i$.
\end{definition}

\begin{claim}
\label{claim:bipartitesoundness}
For each $C \in {[n] \choose q - s}$ and $p \in P_s$, the matrix $A_{C,p}$ defined in \cref{def:bipartitekikuchi} satisfies the following properties.
\begin{enumerate}[(1)]
\item The matrix $A_{C,p}$ has exactly $D = {q - s \choose \frac{q-1}{2}} {n - (q - s) \choose \ell - \frac{q - 1}{2}}{\abs{P_s} - 1 \choose \ell}$ nonzero entries.
\item For each $x \in \Fits^n$ and $y \in \Fits^{P_s}$, let $z \in \Fits^{V_L}$ and $w \in \Fits^{V_R}$ be defined as: $z_{S_1, S_2} := \prod_{v \in S_1} x_v \prod_{p \in S_2} y_p$ and $w_{T_1, T_2} := \prod_{v \in T_1} x_v \prod_{p \in T_2} y_p$. Then, $z^{\top} A_{C,p} w = D y_p \prod_{v \in C} x_v$.
\end{enumerate}
In particular, $z^{\top} A w = D \Psi(x,y)$. 

As additional notation, we let $N_L \defeq \abs{V_L}$, $N_R \defeq \abs{V_R}$. Finally, we let $d_L$ and $d_R$ denote (upper bounds on) the average left and right degrees of each $A_i$, i.e., $d_L = \frac{\delta n D}{N_L}$ and $d_R = \frac{\delta n D}{N_R}$.
\end{claim}
\begin{proof}
To prove Item (1), we will count the number of edges. Let $C \in {[n] \choose q-s}$, $p \in P_s$. By definition, we have an edge $((S_1, S_2), (T_1, T_2))$ in $G_{C,p}$  iff $\abs{S_1 \cap C} = \frac{q-1}{2}$ and $p \notin S_2$. The number of such $S_1$ is ${q - s \choose \frac{q-1}{2}} {n - (q - s) \choose \ell - \frac{q - 1}{2}}$, and the number of such $S_2$ is ${\abs{P_s} - 1 \choose \ell}$. Hence, Item (1) holds.

To prove Item (2), we observe that
\begin{flalign*}
&z^{\top} A_{C,p} w = \sum_{((S_1, S_2), (T_1, T_2)) \in E(G_{C,p})} z_{S_1, S_2} w_{T_1, T_2} = \sum_{((S_1, S_2), (T_1, T_2)) \in E(G_{C,p})} \prod_{v \in S_1} x_v \prod_{p' \in S_2} y_{p'} \prod_{v \in T_1} x_v \prod_{p' \in T_2} y_{p'} \\
&=\sum_{((S_1, S_2), (T_1, T_2)) \in E(G_{C,p})} \prod_{v \in S_1 \oplus T_1} x_v \prod_{p' \in S_2 \oplus T_2} y_{p'} = \sum_{((S_1, S_2), (T_1, T_2)) \in E(G_{C,p})} y_p \cdot \prod_{v \in C} x_v = D  y_p \cdot \prod_{v \in C} x_v \mper
\end{flalign*}

The ``in particular'' follows immediately from Item (2) and the definition of $A$ and the $A_i$'s.
\end{proof}

The key technical lemma, which we shall prove in \cref{sec:bipartiterowpruning}, shows that we can find an approximately biregular subgraph of $A_i$ for each $i \in [k]$.
\begin{lemma}[Approximately regular submatrix]
\label{lem:bipartiterowpruning}
Let $A_1, \dots, A_k$ be defined as in \cref{def:bipartitekikuchi}. There exists a positive integer $D'$ with $D \geq D' \geq \frac{D}{2}$ such that the following holds. For each $i \in [k]$ and $(C,p) \in H_i$, there exists a matrix $B_{i, C, p} \in \Bits^{V_L \times V_R}$ with the following properties:
\begin{enumerate}[(1)]
\item $B_{i, C, p}$ is a ``subgraph'' of $A_{C,p}$. Namely, if $B_{i, C, p}((S_1, S_2), (T_1, T_2)) = 1$, then $A_{C, p}((S_1, S_2), (T_1, T_2)) = 1$.
\item $B_{i,C,p}$ has exactly $D'$ nonzero entries.
\item The matrix $B_i \defeq \sum_{(C,p) \in H_i} B_{i, C, p}$ has at most $O(d_L)$ nonzero entries per row and $O(d_R)$ nonzero entries per column.
\end{enumerate}
\end{lemma}

With \cref{lem:bipartiterowpruning} in hand, we can now finish the proof. Let $B_i$ be the matrix defined in \cref{lem:bipartiterowpruning}. We observe that Items (1), (2), and (3) in \cref{lem:bipartiterowpruning}, along with \cref{claim:bipartitesoundness}, imply that for each $x \in \Fits^n$ and $y \in \Fits^{P_s}$, there exist $z \in \Fits^{V_L}$ and $w \in \Fits^{V_R}$ such that $z^{\top} B w = D' \Psi(x,y)$. Hence, for any $x \in \Fits^n$ and $y \in \Fits^{P_s}$, it holds that $D' \Psi(x,y) \leq \norm{B}_2 \cdot \sqrt{N_L N_R}$. We thus have that $\E_b[\val(\Psi)] \leq \frac{\sqrt{N_L N_R}}{D'} \E_b[\norm{B}_2]$.

It remains to bound $\E_b[\norm{B}_2]$, which we do using Matrix Khintchine (\cref{fact:matrixkhintchine}) in the following claim.
\begin{claim}
\label{claim:specnormbipartite}
Let $B$ be the matrix defined in \cref{lem:bipartiterowpruning}. Then, $\E_b[\norm{B}_2] \leq O(\sqrt{d_L d_R k \ell \log n})$.
\end{claim}

We postpone the proof of \cref{claim:specnormbipartite} to the end of this section, and finish the proof of \cref{thm:bipartiteref}. We have that
\begin{flalign*}
&\E_b[\val(\Psi)] \leq \frac{\sqrt{N_L N_R}}{D'} \E_b[\norm{B}_2]  \leq  \frac{2\sqrt{N_L N_R}}{D}O(\sqrt{d_L d_R k \ell \log n}) \\
&= O(1) \cdot \sqrt{ \frac{ N_L d_L N_R d_R k \ell \log n}{D^2}} = \delta n O(\sqrt{k \ell \log n}) \mcom
\end{flalign*}
as required, where we recall that $d_L = \delta n D/N_L$ and $d_R = \delta n D/N_R$.

We now finish the proof of \cref{claim:specnormbipartite}.
\begin{proof}[Proof of \cref{claim:specnormbipartite}]
 By Matrix Khintchine (\cref{fact:matrixkhintchine}), we have $\E_b[\norm{B}_2] \leq O(\sqrt{\sigma^2 \log (N_L + N_R)})$, where $\sigma^2 = \max (\norm{\sum_{i = 1}^k B_i B_i^{\top}}_2, \norm{\sum_{i = 1}^k B_i^{\top} B_i}_2)$. Since $B_i B_i^{\top}$ is symmetric, $\norm{B_i B_i^{\top}}_2$ is bounded by the maximum $\ell_1$-norm of a row in this matrix. We observe that the $\ell_1$-norm of the $(S_1, S_2)$-th row in $B_i B_i^{\top}$ is simply the number of length $2$ walks starting from the left vertex $(S_1, S_2)$ in the bipartite graph with adjacency matrix $B_i$. As this graph has maximum left degree $O(d_L)$ and maximum right degree $O(d_R)$, it follows that this is at most $O(d_L d_R)$. Similarly, the maximum $\ell_1$-norm of a row in $B_i^{\top} B_i$ is the number of length $2$ walks starting from the right vertex $(T_1, T_2)$ in the bipartite graph $B_i$, and this is at most $O(d_R d_L)$. Hence, $\norm{\sum_{i = 1}^k B_i B_i^{\top}}_2 \leq \sum_{i = 1}^k O(d_L d_R) = O(k d_L d_R)$, and similarly $\norm{\sum_{i = 1}^k B_i^{\top} B_i}_2 \leq O(k d_L d_R)$ as well.
 
 We can thus set $\sigma^2 = O(k d_L d_R)$ and apply \cref{fact:matrixkhintchine} to conclude that $\E_b[\norm{B}_2] \leq O(\sqrt{k d_L d_R \log (N_L + N_R)})$. Recall that we have $N_L = {n \choose \ell} {\abs{P_s} \choose \ell} \leq n^{\ell} (nk)^{\ell} \leq n^{O(\ell)}$ and $N_R = {n \choose \ell + 1 - s} {\abs{P_s} \choose \ell + 1} \leq n^{\ell} (nk)^{\ell + 1} \leq n^{O(\ell)}$. Hence, $\log (N_L + N_R) = O(\ell \log n)$, which finishes the proof.
 \end{proof}

\subsection{Finding an approximately regular subgraph: proof of \cref{lem:bipartiterowpruning}}
\label{sec:bipartiterowpruning}

In this section, we prove \cref{lem:bipartiterowpruning}. Similar to \cref{lem:regularrowpruning}, we will prove \cref{lem:bipartiterowpruning} by using the strategy, due to~\cite{Yankovitz24}, of bounding ``conditional first moments''. 
\begin{lemma}[Conditional first moment bounds]
\label{lem:bipartitemoments}
Fix $i \in [k]$. For a left vertex $(S_1, S_2)$, let $\deg_{i,L}(S_1, S_2)$ denote the degree of $(S_1, S_2)$ in $A_i$, and for a right vertex $(T_1, T_2)$, let $\deg_{i,R}(T_1, T_2)$ denote the right degree in $A_i$.

Let $(C,p) \in H_i$, and let $\mu_{L,C,p}$ denote the distribution over left vertices that first chooses a uniformly random edge in $A_{C,p}$ and outputs its left endpoint. Similarly, let $\mu_{R,C,p}$ denote the distribution that outputs the right endpoint. Then, it holds that
\begin{flalign*}
&\E_{(S_1, S_2) \sim \mu_{L, C,p}}[\deg_{i,L}(S_1, S_2)] \leq 1 + O(1) \left(\frac{\ell}{n}\right)^{\frac{q - 1}{2}} \delta n\\
&\E_{(T_1, T_2) \sim \mu_{R, C,p}}[\deg_{i,R}(T_1, T_2)] \leq 1 + O(1) \left( \frac{\ell}{n} \right)^{\frac{q + 1}{2} - s} \frac{\ell}{\abs{P_s}} \delta n\mper
\end{flalign*}
\end{lemma}

\begin{claim}[Degree bound]
\label{claim:bipartitedegbound}
Let $d_L$ and $d_R$ be the quantities defined in \cref{def:bipartitekikuchi}. Then, for the choice of parameters given in \cref{thm:bipartiteref}, it holds that
\begin{flalign*}
&d_L \geq \Omega\left(\left(\frac{\ell}{n}\right)^{\frac{q - 1}{2}} \delta n\right) \text { and } \left(\frac{\ell}{n}\right)^{\frac{q - 1}{2}} \delta n \geq 1 \\ 
&d_R \geq \Omega\left( \left( \frac{\ell}{n} \right)^{\frac{q + 1}{2} - s} \frac{\ell}{\abs{P_s}} \delta n\right) \text{ and }  \left( \frac{\ell}{n} \right)^{\frac{q + 1}{2} - s} \frac{\ell}{\abs{P_s}} \delta n \geq 1 \mper
\end{flalign*}
\end{claim}

We postpone the proofs of \cref{lem:bipartitemoments,claim:bipartitedegbound} to the end of this subsection, and now use them to finish the proof of \cref{lem:bipartiterowpruning}.

\begin{proof}[Proof of \cref{lem:bipartiterowpruning} from \cref{lem:bipartitemoments,claim:bipartitedegbound}]
Fix $i \in [k]$. Let $\Gamma$ be a constant (to be chosen later), and let $V'_L = \{(S_1, S_2) : \deg_{i,L}(S_1, S_2) \leq \Gamma d_L\}$, and let $V'_R = \{(T_1, T_2) : \deg_{i,R}(T_1, T_2) \leq \Gamma d_R\}$. 

Let $(C,p) \in H_i$. We let $A'_{i,C,p}$ be the matrix where $A'_{i,C,p}((S_1, S_2), (T_1, T_2)) = A_{C,p}((S_1, S_2), (T_1, T_2))$ if $(S_1, S_2) \in V'_L$ and $(T_1, T_2) \in V'_R$, and otherwise $A'_{i, C, p}((S_1, S_2), (T_1, T_2)) =  0$. Namely, we have ``zeroed out'' all rows of $A_{C,p}$ that are not in $V'_L$ and all columns that are not in $V'_R$.

The left degree conditional moment bound from \cref{lem:bipartitemoments}, combined with the lower bound on $d_L$ from \cref{claim:bipartitedegbound} implies that $\E_{(S_1, S_2) \sim \mu_{L, C,p}}[\deg_{i,L}(S_1, S_2)] \leq O(d_L)$. Similarly, we have $\E_{(T_1, T_2) \sim \mu_{R, C,p}}[\deg_{i,R}(T_1, T_2)] \leq O(d_R)$. Hence, applying Markov's inequality, the number of left vertices $(S_1, S_2)$ that are adjacent to an edge labeled by $(C,p)$ and have $\deg_{i,L}(S_1, S_2) > \Gamma d_L$ is at most $O(D/\Gamma)$. Similarly, the number of right vertices $(T_1, T_2)$ that are adjacent to an edge labeled by $(C,p)$ and have $\deg_{i,R}(T_1, T_2) > \Gamma d_R$ is at most $O(D/\Gamma)$. Hence, there must be at least $D (1 - O(1/\Gamma))$ edges, i.e., nonzero entries, in $A'_{i,C,p}$.

Now, we let $B_{i, C, p}$ be any subgraph of $A'_{i,C,p}$ where $B_{i,C,p}$ has \emph{exactly} $D' = \floor{D(1 - O(1/\Gamma))}$ edges. This can be achieved by simply removing edges if there are too many. By choosing $\Gamma$ to be a sufficiently large constant, we ensure that $D' \geq D/2$.

To prove the third property, we observe that for any left vertex $(S_1, S_2)$, the matrix $B_i = \sum_{(C,p) \in H_i} B_{i,C,p}$ has at most $\Gamma d_L = O(d_L)$ nonzero entries in the $(S_1, S_2)$-th row. Indeed, this follows because it is a subgraph of the original graph $A_i$, and if $(S_1, S_2)$ had degree $> \Gamma d_L$ in $A_i$ then it has degree $0$ in $B_i$. Similarly, any right vertex $(T_1, T_2)$ has degree at most $O(d_R)$ in the matrix $B_i$. This finishes the proof.\end{proof}

It remains to prove \cref{lem:bipartitemoments,claim:bipartitedegbound}, which we do now.
\begin{proof}[Proof of \cref{lem:bipartitemoments}]
Let $(C,p) \in H_i$ be an edge. We will first compute the left degree and then the right degree. We observe that for any $(S_1, S_2) \in V_L$, the vertex $(S_1, S_2)$ is adjacent to at most one edge labeled by $(C,p)$. Hence, it follows that $\mu_{L, C,p}$ is uniform over pairs of sets $(S_1, S_2)$ such that $\abs{S_1 \cap C} = \frac{q - 1}{2}$ and $p \notin S_2$. We have
\begin{flalign*}
&\E_{(S_1, S_2) \sim \mu_{L, C,p}}[\deg_{i,L}(S_1, S_2)] = 1 + \frac{1}{D} \sum_{(C', p') \in H_i \setminus \{(C,p)\}} \abs{\{(S_1, S_2) : \abs{S_1 \cap C} = \abs{S_1 \cap C'} = \frac{q - 1}{2} \text{ and } p, p' \notin S_2\}} \\
&\leq 1 + \frac{\delta n - 1}{D} {q-s \choose \frac{q-1}{2}}^2 {n - 2(q - s) \choose \ell - (q-1)} {\abs{P_s} - 2 \choose \ell} \mcom
\end{flalign*}
where we use that $\abs{H_i} \leq \delta n$.

Applying \cref{fact:binomest} and using that $\abs{P_s} \geq 4 \ell$, we have
\begin{flalign*}
&\frac{1}{D} {q-s \choose \frac{q-1}{2}}^2 {n - 2(q - s) \choose \ell - (q-1)} {\abs{P_s} - 2 \choose \ell} = \frac{1}{{q - s \choose \frac{q-1}{2}} {n - (q - s) \choose \ell - \frac{q - 1}{2}}{\abs{P_s} - 1 \choose \ell}} {q-s \choose \frac{q-1}{2}}^2 {n - 2(q - s) \choose \ell - (q-1)} {\abs{P_s} - 2 \choose \ell} \\
&\leq O(1) \left(\frac{\ell}{n}\right)^{\frac{q - 1}{2}} \frac{1}{ {\abs{P_s} - 1 \choose \ell}} {\abs{P_s} - 2 \choose \ell}  \leq O(1) \left(\frac{\ell}{n}\right)^{\frac{q - 1}{2}} \mper
\end{flalign*}
Hence, $\E_{(S_1, S_2) \sim \mu_{L, C,p}}[\deg_{i,L}(S_1, S_2)] \leq 1 + O(1) \left(\frac{\ell}{n}\right)^{\frac{q - 1}{2}} \delta n$. 

We now compute $\E_{(T_1, T_2) \sim \mu_{R, C,p}}[\deg_{i,R}(T_1, T_2)]$. As before, for any $(T_1, T_2) \in V_R$, the vertex $(T_1, T_2)$ is adjacent to at most one edge labeled by $(C,p)$. So, it follows that $\mu_{R,C,p}$ is uniform over pairs of sets $(T_1, T_2)$ such that $\abs{T_1 \cap C} = \frac{q + 1}{2} - s$ and $p \in T_2$. We have
\begin{flalign*}
&\E_{(T_1, T_2) \sim \mu_{R, C,p}}[\deg_{i,R}(T_1, T_2)] = 1 + \frac{1}{D} \sum_{(C', p') \in H_i \setminus \{(C,p)\}} \abs{\{(T_1, T_2) : \abs{T_1 \cap C} = \abs{T_1 \cap C'} = \frac{q + 1}{2} - s  \text{ and } p, p' \in T_2\}} \\
&\leq 1 + \frac{\delta n - 1}{D} {q-s \choose \frac{q+1}{2} - s}^2 {n - 2(q - s) \choose (\ell + 1 - s) - (q + 1) + 2s} {\abs{P_s} - 2 \choose \ell - 1}
\end{flalign*}
We have
\begin{flalign*}
 &\frac{1}{D} {q-s \choose \frac{q+1}{2} - s}^2 {n - 2(q - s) \choose (\ell + 1 - s) - (q + 1) + 2s} {\abs{P_s} - 2 \choose \ell - 1} \\
 &=  \frac{1}{{q - s \choose \frac{q-1}{2}} {n - (q - s) \choose \ell - \frac{q - 1}{2}}{\abs{P_s} - 1 \choose \ell}}  {q-s \choose \frac{q+1}{2} - s}^2 {n - 2(q - s) \choose (\ell + 1 - s) - (q + 1) + 2s} {\abs{P_s} - 2 \choose \ell - 1} \\
 &= \frac{1}{ {n - (q - s) \choose \ell - \frac{q - 1}{2}}{\abs{P_s} - 1 \choose \ell}}  {q-s \choose \frac{q+1}{2} - s} {n - 2(q - s) \choose \ell - q + s} {\abs{P_s} - 2 \choose \ell - 1} \\
 &\leq O(1) \left( \frac{\ell}{n} \right)^{\frac{q + 1}{2} - s} \frac{\ell}{\abs{P_s}} \mcom
\end{flalign*}
where the last inequality is by \cref{fact:binomest} and uses that $\abs{P_s} \geq 4 \ell$.
Hence, $\E_{(T_1, T_2) \sim \mu_{R, C,p}}[\deg_{i,R}(T_1, T_2)]  \leq 1 + O(1) \left( \frac{\ell}{n} \right)^{\frac{q + 1}{2} - s} \frac{\ell}{\abs{P_s}} \delta n$. 
\end{proof}

\begin{proof}[Proof of \cref{claim:bipartitedegbound}]
We observe that by \cref{fact:binomest} and that $\abs{P_s} \geq 4 \ell$,
\begin{flalign*}
\frac{\delta n D}{N_L} = \delta n \cdot \frac{{q - s \choose \frac{q-1}{2}} {n - (q - s) \choose \ell - \frac{q - 1}{2}}{\abs{P_s} - 1 \choose \ell}} {{n \choose \ell} {\abs{P_s} \choose \ell}} \geq \delta n \cdot \Omega(1) \left(\frac{\ell}{n}\right)^{\frac{q - 1}{2}} \mper
\end{flalign*}
Because $\ell = n^{1 - 2/q} \delta^{-2/q}$, we have $\left(\frac{\ell}{n}\right)^{\frac{q - 1}{2}} \delta n \geq 1$, which finishes the case for the left degree.

We also have that
\begin{flalign*}
\frac{\delta n D}{N_R} = \delta n \cdot \frac{{q - s \choose \frac{q-1}{2}} {n - (q - s) \choose \ell - \frac{q - 1}{2}}{\abs{P_s} - 1 \choose \ell}} {{n \choose \ell + 1 - s} {\abs{P_s} \choose \ell + 1}} \geq \delta n \cdot \Omega(1) \left(\frac{\ell}{n}\right)^{\frac{q + 1}{2} - s} \frac{\ell}{\abs{P_s}} \mper
\end{flalign*}
Now, because $\abs{P_s} \leq nk/d_s$ where $d_s = \left(\frac{\ell}{n}\right)^{s - \frac{3}{2}} k$ and $\ell = n^{1 - 2/q} \delta^{-2/q}$, it follows that $\delta n \left(\frac{\ell}{n}\right)^{\frac{q + 1}{2} - s} \frac{\ell}{\abs{P_s}} \geq \delta n \left(\frac{\ell}{n}\right)^{\frac{q + 1}{2} - s} \frac{\ell}{n} \left(\frac{\ell}{n}\right)^{s - \frac{3}{2}} = \delta n  \left(\frac{\ell}{n}\right)^{\frac{q}{2}} \geq 1$. This finishes the case for the right degree.
\end{proof}
\section*{Acknowledgements}
We thank Shachar Lovett, Raghu Meka, Lisa Sauermann, and Ola Svensson for organizing a wonderful workshop at the Bernoulli Center for Fundamental Studies at EPFL on the synergies of combinatorics and theoretical computer science that led to this paper. We also thank the anonymous FOCS 2025 reviewers for their insightful comments that have improved the presentation of the paper.

\bibliographystyle{alpha}
\bibliography{references.bib}

\end{document}